\documentclass{article}
\usepackage[utf8]{inputenc}
\usepackage[left=25mm, right=25mm, top=20mm, bottom=20mm]{geometry}
\usepackage{amsfonts}
\usepackage{amsmath}
\usepackage{amssymb}
\usepackage{amsthm}
\usepackage{authblk}
\usepackage{bm,bbm}
\usepackage{booktabs}
\usepackage{cancel}
\usepackage{caption}
\usepackage{color}
\usepackage[radius=.8mm,edge= 5mm,mark=o]{dynkin-diagrams}
\usepackage[mathscr]{euscript}
\usepackage{hyperref}
\usepackage{physics}
\usepackage{stmaryrd}
\hypersetup{colorlinks,breaklinks,
            linkcolor = purple,
            citecolor = teal,
            urlcolor = purple,
            }
\definecolor{darkblue}{rgb}{0.1,0.1,0.7}
\definecolor{darkred}{rgb}{0.5,0.1,0.1}
\definecolor{darkgreen}{rgb}{0.0,0.42,0.06}

\newcommand{\ee}{\mathrm{e}}
\newcommand{\ii}{\mathrm{i}}

\def\t{\mathbf{t}}
\def\M{\overline{\mathcal{M}}}
\def\C{\mathbb{C}}

\theoremstyle{definition}
\newtheorem{theorem}{Theorem}[section]

\newtheorem{lemma}[theorem]{Lemma}

\theoremstyle{remark}

\numberwithin{equation}{section}

\title{Revisiting three-dimensional GLSMs and K-theoretic I-functions}
\author[1,2]{Taro Kimura}
\author[1]{Yinda Li}
\author[3]{Xiaohan Yan}
\affil[1]{Université Bourgogne Europe, CNRS, IMB UMR 5584, 21000 Dijon, France}
\affil[2]{Institut Universitaire de France (IUF), France}
\affil[3]{Institut de Mathématiques de Toulouse, UMR 5219, Université de Toulouse, CNRS, 31400 Toulouse, France}
\date{}

\begin{document}

\maketitle

\begin{abstract}
We revisit the correspondence between three-dimensional gauged linear sigma models and K-theoretic Gromov--Witten theory. 
We show that the K-theoretic Hori--Vafa formula for the grassmannian arises from the three-dimensional gauge theory upon including an appropriate anomaly-cancellation factor.   
Furthermore, we derive K-theoretic I-functions for isotropic grassmannians by incorporating symmetric and anti-symmetric tensor representation fields together with suitable boundary conditions.
\end{abstract}

\setcounter{tocdepth}{2}
\tableofcontents

\section{Introduction and Summary}

Exact results obtained by supersymmetric localization have revealed profound connections between supersymmetric gauge theories and enumerative geometry. A particularly striking example is provided by two-dimensional $\mathcal{N}=(2,2)$ gauged linear sigma models (GLSMs)~\cite{Witten:1993yc}, whose partition functions on the two-sphere~\cite{Benini:2012ui,Doroud:2012xw} encode nontrivial geometric information about the target Kähler manifold~\cite{Jockers:2012dk}. The vortex partition functions obtained from the factorization of the two-sphere partition function, which are also derived directly from the hemisphere partition functions~\cite{Honda:2013uca,Hori:2013ika} and the two-dimensional $\Omega$-background~\cite{Fujimori:2015zaa}, have been identified with Givental's I-functions~\cite{Bonelli:2013mma}, while the corresponding J-functions encode genus-zero Gromov--Witten (GW) theory~\cite{Givental:1996heu} and are related to the I-functions through the mirror map~\cite{Givental:1998}. These developments have established a direct bridge between supersymmetric gauge theories, mirror symmetry, and quantum cohomology. This correspondence suggests that supersymmetric partition functions provide a natural physical realization of Givental's formalism.

A natural extension of this picture is obtained by lifting the two-dimensional theory to the three-dimensional theory compactified on a circle. 
The hemisphere partition function admits a three-dimensional counterpart given by the partition function on the solid torus $\mathbb{D}^2 \times \mathbb{S}^1$~\cite{Yoshida:2014ssa}. 
Supersymmetric localization expresses this partition function in terms of K-theoretic vortex partition functions, providing a physical realization of quantum K-theory~\cite{Givental:2001clq}. 
Correspondingly, the mathematical construction admits K-theoretic analogues, leading to K-theoretic I-functions as $q$-deformations of their cohomological counterparts. 
This correspondence establishes a direct connection between three-dimensional gauge theories and quantum K-theory~\cite{Yoshida:2014ssa,Jockers:2018sfl}.
Existing constructions are available for grassmannians~\cite{Ueda:2019qhg,Jockers:2019lwe,Givental:2021dbg} and flag varieties~\cite{Yan2024} through non-abelian localization.
See also \cite{Pushkar:2016qvw,Koroteev:2017nab} for related developments on the so-called vertex functions of grassmannians and flag varieties.

\subsubsection*{Anomaly, factorization, and Hori--Vafa formula}

A fundamental result of Hori and Vafa~\cite{Hori:2000kt} predicts that the genus-zero GW theory of the grassmannian $\operatorname{Gr}(k,n)$ can be reconstructed from that of $(\mathbb{P}^{n-1})^k$.
Concretely, the Hori--Vafa (HV) formula expresses the grassmannian I-function as an anti-symmetrization of the I-functions of projective spaces.
Via the mirror theorem, this yields a corresponding description of the grassmannian J-function, and provides an explicit computational framework for its quantum cohomology.

In the K-theoretic setting, on the other hand, a naive generalization of the HV formula is not available.
From a gauge theory perspective, the obstruction is understood as a consequence of anomaly induced by the Chern--Simons (CS) term, which is specific to three-dimensional gauge theories.
Even in the absence of a bare CS term, quantum effects generate an effective CS term, and the resulting GLSM partition function does not factorize into those of projective spaces.

In this paper, we incorporate an additional boundary term to compensate for the $\mathrm{U}(1)$ anomaly. 
This leads to a factorization of the grassmannian partition function, which provides a proof of the HV formula in the K-theoretic setting. 
As an immediate consequence, we obtain a determinantal formula for the grassmannian partition function $\operatorname{Gr}(k,n)$, which is a rank-$k$ $q$-Casoratian minor of solutions to the rank-$n$ $q$-difference equation, i.e., the $q$-hypergeometric functions.
Moreover, such a factorization property is compatible with the Wilson loop insertion along $\mathbb{S}^1$ direction. 
In the Coulomb branch localization formula, the Wilson loop contribution is given by the character polynomial, i.e., Schur polynomial.
We show that the Schur polynomial of the $q$-difference operators generates the Wilson loop, and the resulting integral is again given by a determinant of the $q$-hypergeometric functions. 
We compute the two-point function of the Wilson loops and address the underlying algebraic structure in relation to the quantum K-theory.

\subsubsection*{Isotropic grassmannians}

Whereas the I/J-functions for the grassmannians are well established, analogous K-theoretic formulas for isotropic grassmannians have not yet appeared in the literature. 
Another purpose of this paper is to fill this gap by deriving the K-theoretic I-functions for symplectic and orthogonal grassmannians directly from the localization of the corresponding three-dimensional gauge theories.

The ordinary grassmannian $\operatorname{Gr}(k,n)$ admits a well-known GLSM realization as a $\mathrm{U}(k)$ gauge theory with $n$ fundamental chiral multiplets.
The symplectic and orthogonal grassmannians, denoted by $\operatorname{SG}(k,2n)$ and $\operatorname{OG}(k,2n+\chi)$ with $\chi \in \{0,1\}$, respectively, are isotropic subvarieties of the ordinary grassmannian, where their GLSM realizations were explored by Gu, Sharpe, and Zou~\cite{Gu:2020oeb}, following a proposal by Okonek and Teleman~\cite{Okonek:2019lwg}. 
In particular, the symplectic grassmannian $\operatorname{SG}(k,2n)$ is realized as a $\mathrm{U}(k)$ gauge theory with $2n$ fundamental chiral multiplets and one chiral multiplet in the anti-symmetric representation, while the orthogonal grassmannian $\operatorname{OG}(k,2n+\chi)$ is realized as $\mathrm{U}(k)$ gauge theories with $2n+\chi$ fundamental chiral multiplets and one chiral multiplet in the symmetric representation. 
There is a corresponding superpotential that enforces the isotropic condition, and the resulting Higgs branch of the GLSM reproduces the symplectic and orthogonal grassmannians. 
In this paper, we compute the partition functions of these GLSMs on $\mathbb{D}^2\times \mathbb{S}^1$ using supersymmetric localization, and show that they reproduce the K-theoretic I-functions of $\operatorname{SG}(k,2n)$ and $\operatorname{OG}(k,2n+\chi)$ obtained from the quantum Lefschetz theorem, which generalizes the quantum cohomology computation by Bertram, Ciocan-Fontanine, and Kim~\cite{Bertram:2004}.

One more ingredient of the localization computation on $\mathbb{D}^2 \times \mathbb{S}^1$ is the choice of boundary conditions for the chiral multiplets. 
For the ordinary grassmannian, we assign the Neumann boundary condition to all fundamental chiral multiplets. 
Since the isotropic grassmannians are realized as subvarieties of the ordinary grassmannian, we still assign the Neumann boundary condition to the fundamental chiral multiplets. 
On the other hand, we propose to apply the Dirichlet boundary condition to the anti-symmetric and symmetric chiral multiplets. 
These tensor chiral multiplet contributions under the Dirichlet boundary condition indeed reproduce an additional (K-theoretic) Euler class appearing in the quantum Lefschetz theorem. 

\paragraph{Organization of the paper}

This paper is organized as follows.
In \autoref{sec:Gr} we explore the three-dimensional GLSM describing the grassmannians.
Firstly we focus on the projective space GLSM, which is a building block of the higher-rank theory in \autoref{sec:P}.
We study the hemisphere partition function based on the Coulomb branch localization formula and discuss the underlying quantum $\mathscr{D}_q$-module structure.
We then consider the GLSM for $X = \operatorname{Gr}(k,n)$ in \autoref{sec:Gr_sub}.
We introduce an additional boundary term that compensates the U(1) anomaly of the bulk 3d theory.
The resulting partition function admits a determinantal formula, which proves the HV formula in the three-dimensional/K-theoretic setup. 
We see this factorization property at the level of the integral kernel, and also study the Wilson loop insertion.
In \autoref{sec:TGr}, we explore the cotangent bundle theory, which is the $\mathcal{N}=4$ uplift of the grassmannian GLSM.
We show that the $\mathcal{N}=4$ theory realizes the I-functions of the odd tangent bundle and the cotangent bundle depending on the boundary condition.
We then apply the semiclassical analysis in \autoref{sec:Gr_semiclassical}.
We discuss the effective twisted superpotential from the Coulomb branch formula and the associated twisted F-term equations in relation to the quantum K-ring relations.

\autoref{sec:SG} and \autoref{sec:OG} are devoted to the study of the isotropic grassmannians, $\operatorname{SG}(k,2n)$ and $\operatorname{OG}(k,2n+\chi)$.
We consider the Coulomb branch formulas for the three-dimensional GLSM involving the chiral multiplet in the (anti-)symmetric representation, and show that they reproduce the I-function for the symplectic and orthogonal grassmannians. 
We in particular study $\operatorname{OG}(1,2n+\chi) = Q^{2n+\chi-2}$ in \autoref{sec:Q} and discuss the isomorphism $Q^2 \cong \mathbb{P}^1 \times \mathbb{P}^1$ at the level of the I-functions.
We then study in \autoref{sec:TSG} and \autoref{sec:TOG} the cotangent bundle theory.
We show that the I-functions for the odd tangent bundle and the cotangent bundle are obtained from the $\mathcal{N}=4$ theory, depending on the boundary condition as in the previous case.
We apply the semiclassical analysis in \autoref{sec:SG_semiclassical} and \autoref{sec:OG_semiclassical} for SG and OG.
We obtain the twisted F-term equations from the effective twisted superpotential derived from the Coulomb branch formula, which are identified with the equivariant K-ring relations for SG and OG.

\autoref{sec:I-func} is devoted to the summary of mathematical background of the K-theoretic I-functions and \autoref{sec:formulas} is for the summary of the formulas.

\paragraph{Acknowledgment}

The work of TK was supported by EIPHI Graduate School (No.~ANR-17-EURE-0002) and the Bourgogne-Franche-Comté region. The work of XY was supported by the ERC starting grant 101164820 {\em Spin Curves enumeration}.

\section{Grassmannian GLSM}\label{sec:Gr}

In this section, we study the GLSM realization of the grassmannian $\operatorname{Gr}(k,n)$ and its relation to the I-function~\cite{Jockers:2018sfl,Ueda:2019qhg,Jockers:2019lwe,Jockers:2021omw}. The GLSM for $X = \operatorname{Gr}(k,n)$ is given by the $\mathrm{U}(k)$ gauge theory with $n$ fundamental chiral multiplets. 
The equivariant I-function agrees with the Higgs branch formula (i.e., hypergeometric formula) of the partition function of 3d $\mathcal{N} = 2$ gauge theory on $\mathbb{D}^2 \times \mathbb{S}^1$ under the Neumann boundary condition for the fundamental chiral multiplets.
The supersymmetric localization of the partition function on $\mathbb{D}^2 \times \mathbb{S}^1$ yields the contour integral formula, which is known as the Coulomb branch localization formula~\cite{Yoshida:2014ssa},
\begin{align}
    Z_X = \frac{1}{k!} \int \omega_X(x;z_1,\ldots,z_k) \, , \label{eq:Z_Gr}
\end{align}
where the integral kernel is given by
\begin{align}
    \omega_X(x;z_1,\ldots,z_k) = \prod_{1 \le i \neq j \le k} \left(\frac{z_i}{z_j};q\right)_\infty \prod_{i=1}^k \frac{x^{-\log_q z_i} \theta(z_i)^\kappa}{\prod_{\beta=1}^n (z_i/a_\beta;q)_\infty} \frac{\dd{z_i}}{2\pi\ii z_i} \, ,
\end{align}
with\\
\begin{minipage}[t]{.45\textwidth}
\begin{itemize}
    \item $z = (z_i)_{i=1,\ldots,k} \in (\mathbb{C}^\times)^k$: gauge fugacities
    \item $a = (a_\alpha)_{\alpha=1,\ldots,n} \in (\mathbb{C}^\times)^n$: flavor fugacities
    \item $q \in \mathbb{C}^\times$, $|q| < 1$: rotation fugacity\\
\end{itemize}    
\end{minipage}
\begin{minipage}[t]{.6\textwidth}
\begin{itemize}
    \item $x \in \mathbb{C}^\times$: (exponential) Fayet--Iliopoulos (FI) parameter
    \item $\kappa \in \mathbb{Z}$: CS level
\end{itemize}    
\end{minipage}
The factor $k!$ is the order of the symmetric group $\mathfrak{S}_k$, which is the Weyl group of the unitary group $\mathrm{U}(k)$.
We write $(z;q)_\infty = \prod_{m=0}^\infty (1-z q^m)$ for $|q| < 1$, $q = \ee^{-\epsilon}$ with $\operatorname{Re} \epsilon > 0$ and $\theta(z) = \theta(z;q)$. 
See \autoref{sec:formulas} for the summary of the formulas.
Hence, the FI term is given by
\begin{align}
    x^{-\log_q z} = \exp \left( - \frac{\log x \log z}{\log q} \right) = \exp \left( \frac{1}{\epsilon} \log x \log z \right) \, .
\end{align}
We set $|x| < 1$, i.e., $\log |x| < 0$ in the following discussion.

The product of theta functions is the Chern--Simons term, which is a quantum completion of the Gaussian term $q^{-\frac{\kappa}{2}\log_q (-z_i)(\log_q (-z_i) - 1)}$ yielding the same shift relation~\cite{Beem:2012mb}.
It is also interpreted as the elliptic genus of the Fermi (chiral) multiplet for $\kappa > 0$ ($\kappa<0$, resp.) in the fundamental representation of the boundary $\mathcal{N}=(0,2)$ theory. 
We will discuss the determinant representation later in \autoref{sec:anomaly}.

From the localization computation, gauge fugacities take a value in the Cartan torus of the (complexified) gauge group.
In the following, we interpret the Coulomb branch formula as a Mellin--Barnes-type integral associated with the following contours.
Let $\mathcal{C}_\alpha$ be the contour encircling the poles at $a_\alpha q^{-d}$ for $d \in \mathbb{Z}_{\ge 0}$, and let $\mathsf{I} : (1,\ldots,k) \hookrightarrow (1,\ldots,n)$ be an injective map such that $\mathsf{I}(i) \neq \mathsf{I}(j)$ for $i \neq j$.
We write $\mathcal{C}_\mathsf{I} = \prod_{i=1}^k \mathcal{C}_{\mathsf{I}(i)}$.
Each choice of $\mathsf{I}$ corresponds to a Higgs vacuum, and the contour integral associated with $\mathcal{C}_\mathsf{I}$ gives the contribution of the corresponding Higgs vacuum, which is identified with the localized equivariant I-function of $\operatorname{Gr}(k,n)$.
Since the integrand is permutation invariant, for $\sigma \in \mathfrak{S}_k$, the contours $\mathcal{C}_\mathsf{I}$ and $\mathcal{C}_{\mathsf{I} \circ \sigma}$ give the same result.
Hence, in the following, we write
\begin{align}
    Z_{X,\mathsf{I}} = \sum_{\sigma \in \mathfrak{S}_k} \frac{1}{k!} \int_{\mathcal{C}_{\mathsf{I}\circ\sigma}} \omega_X = \int_{\mathcal{C}_\mathsf{I}} \omega_X \, .
\end{align}
The number of distinct Higgs vacua is $\binom{n}{k}$, which agrees with the Euler characteristics of $\operatorname{Gr}(k,n)$.

\subsection{$\operatorname{Gr}(1,n) = \mathbb{P}^{n-1}$}\label{sec:P}

We start with the abelian case $X = \operatorname{Gr}(1,n) = \mathbb{P}^{n-1}$, which will be a building block for the general situation $k > 1$.
In this case, the integral formula of the partition function is given by
\begin{align}
    Z_X = \int \omega_X(x;z) \, , \qquad \omega_X(x;z) = \frac{x^{-\log_q z} \theta(z)^\kappa}{\prod_{\beta = 1}^n (z/a_\beta;q)_\infty} \frac{\dd{z}}{2\pi\ii z} \, .\label{eq:Z_Pn}
\end{align}
We have $n$ choices of the integration contours, corresponding to each Higgs vacuum,
\begin{align}
    Z_{X,\alpha} = \int_{\mathcal{C}_\alpha} \omega_X(x;z) = Z^\text{per}_{X,\alpha} Z^\text{vor}_{X,\alpha} \, ,
\end{align}
where $Z^\text{vor}_{X,\alpha}$ is identified with the vortex partition function, while $Z^\text{per}_{X,\alpha}$ is the perturbative factor, 
\begin{subequations}
    \begin{align}
    Z^\text{per}_{X,\alpha} & = \frac{x^{-\log_q a_\alpha} \theta(a_\alpha)^\kappa}{(q;q)_\infty \prod_{\beta (\neq \alpha)}^n (a_\alpha/a_\beta;q)_\infty} \\
    Z^\text{vor}_{X,\alpha} & = \sum_{d=0}^\infty \frac{x^d q^{-\kappa\binom{d+1}{2}} (-a_\alpha)^{\kappa d}}{\prod_{\beta=1}^n(q^{-1}a_\alpha/a_\beta;q^{-1})_d} = \sum_{d=0}^\infty \left( (q/a_\alpha)^{n-\kappa} A x \right)^d \frac{\left((-1)^d q^{\binom{d}{2}}\right)^{n-\kappa}}{\prod_{\beta = 1}^n (qa_\beta/a_\alpha;q)_d} \, , \quad A = \prod_{\beta=1}^n a_\beta \, . \label{eq:q-hypergeometric_Pn}
\end{align}
\end{subequations}
We emphasize that this infinite series formula is available only when $|x| < 1$, and there is no corresponding Higgs vacuum when the FI parameter is positive, $\log |x| > 0$.
We may rewrite the vortex partition function as the $q$-hypergeometric series~\eqref{eq:q-hypergeom} of base $q^{-1}$ and $q$, from which we see the CS level window, $0 \le \kappa \le n$.
See also \autoref{sec:Gr_semiclassical}.

The partition function obeys the $q$-difference equation.
Let $D = x \partial_x$, and we define the $q$-difference operator,
\begin{align}
    \widehat{\mathscr{A}} = q^{\kappa D} \prod_{\alpha = 1}^n \left( 1 - q^{-D}/a_\alpha \right) - (-1)^\kappa x \, .
\end{align}
The kernel $\omega_X$ satisfies an identity of the form, $\widehat{\mathscr{A}} \omega_X(x;z) = \omega_X(x;qz)-\omega_X(x;z)$.
Upon integration over a cycle for which the $q$-boundary term vanishes, we obtain
\begin{align}
    \widehat{\mathscr{A}} Z_{X,\alpha} = 0 \, , \qquad \alpha = 1, \ldots, n \, .
\end{align}

The $q$-difference equation naturally defines a quantum $\mathscr{D}_q$-module. 
Let $\mathscr{D}_q = \mathbb{C}(x) \langle q^{\pm D} \rangle$, $q^D x = q x q^D$, be the algebra of $q$-difference operators acting on the FI parameter $x \in \mathbb{C}^\times$.
Then, the quantum $\mathscr{D}_q$-module associated with $X = \mathbb{P}^{n-1}$ is defined by
\begin{align}
    \mathscr{M}_X = \mathscr{D}_q / \widehat{\mathscr{A}} \mathscr{D}_q \, .
\end{align}
The partition function itself is not an element of the quotient module. 
It is rather a solution of the corresponding difference equation, namely a section of the solution sheaf,
\begin{align}
\operatorname{Sol}(\mathscr{M}_X) = \operatorname{Hom}_{\mathscr{D}_q}\bigl(\mathscr{M}_X, \mathcal{O}_{\mathbb{C}^\times} \bigr),
\end{align}
where $\mathcal{O}_{\mathbb{C}^\times}$ denotes the sheaf of holomorphic functions on the FI parameter. 
Equivalently, $\operatorname{Sol}(\mathscr{M}_X)$ is the local system of holomorphic solutions of the $q$-difference equation. 
A partition function is therefore a local flat section of the quantum $\mathscr{D}_q$-module.

The Coulomb branch integral formula is indeed regarded as a universal integral representation of this solution sheaf. 
As seen above, the quantum difference equation is already encoded in the integration kernel.
The choice of contour becomes relevant only when selecting a particular solution.
For generic flavor parameters $\{a_\alpha\}_{\alpha=1,\ldots,n}$ with $\kappa = 0$, the $n$ functions $\{Z_{X,\alpha}\}_{\alpha=1,\ldots,n}$ form a basis of local flat sections of the rank-$n$ quantum $\mathscr{D}_q$-module. 

This picture is naturally interpreted as a $q$-difference analogue of the Gauss--Manin construction. The Coulomb branch kernel plays the role of a twisting kernel of $q$-de Rham cohomology, while the contour specifies a $q$-analogue of a Betti cycle. 
By twisted $q$-de Rham cohomology, we mean the $q$-difference analogue of twisted de Rham cohomology associated with the Coulomb branch integral kernel: 
Let $\Omega_X(z) = \omega_X(x,q^{-1}z)/\omega_X(x,z)$ be a $q$-connection. 
In this case, we have
\begin{align}
    \Omega_X(z) = \frac{x (-z/q)^{\kappa}}{\prod_{\beta=1}^n(1 - q^{-1}z/a_\beta)} \, .
\end{align}
We then define the twisted $q$-difference operator, $\nabla_X f(z) = \Omega_X(z) f(q^{-1}z) - f(z)$.
By construction, $\omega_X(x,z) \nabla_X f(z) = \omega_X(x,q^{-1}z) f(q^{-1}z) - \omega_X(x,z)f(z)$, so that $\omega_X\nabla_X f$ is an ordinary $q$-difference of the full integrand.
Therefore, the twisted $q$-de Rham cohomology is obtained by quotienting the space of insertions by $\nabla_X$-exact terms.
The partition function is thus the $q$-period without any insertion, $Z_{X,\alpha} = \int_{\mathcal{C}_\alpha} \omega_X(x;z)$. 
We will discuss the Wilson loop insertion in \autoref{sec:Wilson_loop}.

\subsection{$\operatorname{Gr}(k,n)$}\label{sec:Gr_sub}

We compute the contour integral for the grassmannian $X = \operatorname{Gr}(k,n)$.
For an injective map $\mathsf{I}$, we write $a_i = a_{\mathsf{I}(i)}$, and $a_\mathsf{I} = (a_1,\ldots,a_k) = (a_{\mathsf{I}(1)},\ldots,a_{\mathsf{I}(k)})$.
We also write $|d| = \sum_{i=1}^k d_i$.
Then, we have the Higgs branch formula of the grassmannian GLSM partition function,
\begin{align}
    Z_{X,\mathsf{I}} = \int_{\mathcal{C}_\mathsf{I}} \omega_X = Z^\text{per}_{X,\mathsf{I}} Z^\text{vor}_{X,\mathsf{I}} \, ,
\end{align}
where
\begin{subequations}
\begin{align}
    Z^\text{per}_{X,\mathsf{I}} & = \prod_{1 \le i \neq j \le k} (a_i/a_j;q)_\infty \prod_{i=1}^k \frac{x^{-\log_q a_i} \theta(a_i)^\kappa}{(q;q)_\infty \prod_{\beta (\neq \mathsf{I}(i))}^n (a_i/a_\beta;q)_\infty} \\
    Z^\text{vor}_{X,\mathsf{I}} & = \sum_{d_1,\ldots,d_k \ge 0} x^{|d|} \prod_{1 \le i \neq j \le k} \frac{\left( q^{-d_i+d_j} a_i/a_j;q\right)_{\infty}}{\left(a_i/a_j;q\right)_{\infty}} \prod_{i=1}^k \frac{q^{-\kappa\binom{d_i+1}{2}} (-a_i)^{\kappa d_i}}{\prod_{\beta=1}^n (q^{-1} a_i / a_\beta;q^{-1})_{d_i}} \, . \label{eq:Gr_I-func0}
\end{align}
\end{subequations}
By Lemma~\ref{lem:q-shifted_ratio}, we may write
\begin{align}
     \frac{(q^{-d_i+d_j}a_i/a_j;q)_\infty}{(a_i/a_j;q)_\infty} = \frac{\prod_{\ell=-\infty}^{d_i-d_j}(1-q^{-\ell} a_i/a_j)}{\prod_{\ell=-\infty}^{0}(1-q^{-\ell} a_i/a_j)} \, ,
\end{align}
which is a finite product.
Then, we have an alternative form of the vortex partition function
\begin{align}
    Z^\text{vor}_{X,\mathsf{I}}  = \sum_{d_1,\ldots,d_k \ge 0} x^{|d|} \prod_{1 \le i \neq j \le k} \frac{\prod_{\ell = -\infty}^{d_i - d_j} (1 - q^{-\ell} a_i/a_j)}{\prod_{\ell = -\infty}^{0} (1 - q^{-\ell} a_i/a_j)} \prod_{i=1}^k \frac{q^{-\kappa \binom{d_i+1}{2}} (-a_i)^{\kappa d_i}}{\prod_{\beta = 1}^n \prod_{\ell=1}^{d_i} (1 - q^{-\ell} a_i/a_\beta)} \, . \label{eq:Gr_vortex-func} 
\end{align}
We see that the vortex partition function $(1-q)Z^\text{vor}_{X,\mathsf{I}}$ agrees with the equivariant I-function of $\operatorname{Gr}(k,n)$ with the level structure~\eqref{eq:I-func_Gr_V*}.
More precisely, the Chern roots of the tautological bundle $V$ denoted by $(P_1,\ldots,P_k)$ are replaced by the equivariant parameters $(a_{\mathsf{I}(1)},\ldots,a_{\mathsf{I}(k)})$ since we applied the localization. 
Moreover, we need the replacement $q \mapsto q^{-1}$ (and also $x \mapsto (-1)^\kappa x$).
Such a replacement is observed in the context of the quasimap count in the relation to Pandharipande--Thomas invariants~\cite{Pandharipande:2007sq,Nekrasov:2014nea} and also the mirror map~\cite{Ruan2022}.
The current convention of the vortex partition function reproduces the I-function with the level structure with respect to the dual tautological bundle $V^\vee$.
By replacing $\theta(z)$ by $\theta(z^{-1})$ in the integrand, we instead obtain the level structure with respect to $V$.

\subsubsection{Anomaly cancellation and factorization}\label{sec:anomaly}

We write $A_\mathsf{I} = \prod_{i=1}^k a_i$.
By Lemma~\ref{lem:q-shifted_ratio2}, we have
\begin{align}
    \prod_{1 \le i \neq j \le k} \frac{\left( q^{-d_i+d_j} a_i/a_j;q\right)_{\infty}}{\left(a_i/a_j;q\right)_{\infty}} 
    & = \prod_{1 \le i < j \le k} \left( - \frac{a_i}{a_j} \right)^{d_i - d_j} q^{-\binom{d_i-d_j}{2}} \frac{1 - (a_i/a_j) q^{-d_i+d_j}}{1 - a_i/a_j} \nonumber \\
    & = A_\mathsf{I}^{-|d|} \prod_{i=1}^k a_i^{kd_i} q^{-(k-1)\binom{d_i}{2}} \prod_{1 \le i < j \le k} q^{d_i d_j} \frac{a_i q^{-d_i} - a_j q^{-d_j}}{a_i - a_j} \nonumber \\
    & = A_\mathsf{I}^{-|d|} \prod_{i=1}^k a_i^{kd_i} q^{-(k-1)\binom{d_i}{2}} \left( \prod_{1 \le i < j \le k} q^{d_i d_j} \right) \frac{\det_{1 \le i, j \le k} (a_i q^{-d_i})^{k-j}}{\det_{1 \le i, j \le k} a_i^{k-j}} \, .
\end{align}
Here, the cross term $\prod_{1 \le i < j \le k} q^{d_i d_j}$ is an obstruction for the factorization of the partition function as pointed out in \cite{Jockers:2019lwe}.

In order to see a complete factorization of the partition function, we insert an additional theta term, interpreted as the CS term associated with the determinant representation coupled to the U(1) part of the gauge group,
\begin{align}
    Z_{X}^\theta = \frac{1}{k!} \int \omega_X^\theta(x;z_1,\ldots,z_k) \, , \label{eq:Z_Gr_theta}
\end{align}
where the integral kernel is given by
\begin{align}
    \omega_X^\theta(x;z_1,\ldots,z_k) & = \theta(\det z) \omega_X(x;z_1,\ldots,z_k) \nonumber \\
    & = \theta(\det z) \prod_{1 \le i \neq j \le k} \left(\frac{z_i}{z_j};q\right)_\infty \prod_{i=1}^k \frac{x^{-\log_q z_i} \theta(z_i)^\kappa}{\prod_{\beta=1}^n (z_i/a_\beta;q)_\infty} \frac{\dd{z_i}}{2\pi\ii z_i} \, , \label{eq:omega_Gr_theta}
\end{align}
where we write $\det z = z_1 \cdots z_k$.
Then, evaluating the residues, we obtain
\begin{align}
    Z_{X,\mathsf{I}}^\theta & = \prod_{i=1}^k \frac{x^{-\log_q a_i} \theta(a_i)^\kappa}{(q;q)_\infty \prod_{\beta (\neq \mathsf{I}(i))}^n (a_i/a_\beta;q)_\infty} \nonumber \\
    & \qquad \times \sum_{d_1,\ldots,d_k \ge 0} x^{|d|} \theta(q^{-|d|}A_\mathsf{I}) \prod_{1 \le i \neq j \le k} \left( q^{-d_i+d_j} a_i/a_j;q\right)_{\infty} \prod_{i=1}^k \frac{q^{-\kappa\binom{d_i+1}{2}} (-a_i)^{\kappa d_i}}{\prod_{\beta=1}^n (q^{-1} a_i / a_\beta;q^{-1})_{d_i}} \, .
\end{align}
In order to rewrite the partition function in a concise form, we define the elliptic Vandermonde factor of $k$ variables (type $\widehat{A}_{k-1}$),
\begin{align}
    \mathsf{W}_k(z_1,\ldots,z_k) := \prod_{1 \le i < j \le k} z_j \theta(z_i/z_j) \, ,
\end{align}
and 
\begin{align}
    \Theta(z) = \Theta(z_1,\ldots,z_k) := \theta(\det z) \mathsf{W}_k(z_1,\ldots,z_k) = \frac{(q^k;q^k)_\infty^k}{(q;q)_\infty^k} \det_{1 \le i, j \le k} \left( z_i^{j-1} \theta((-1)^{k-1}q^{j-1}z_i^k;q^k) \right) \, .
\end{align}
This determinantal formula was shown by Rosengren and Schlosser~\cite[Proposition 6.1]{Rosengren2006}.
By the shift relation~\eqref{eq:theta_shift}, we compute
\begin{align}
    \Theta(\{q^{-d_i} z_i\}_{i=1}^k) 
    = \Theta(z) \prod_{i=1}^k (-z_i)^{k d_i} q^{-k \binom{d_i+1}{2}} \, . \label{eq:q-const}
\end{align}
Moreover, by definition of the theta function~\eqref{eq:theta_def}, we have
\begin{align}
    \prod_{1 \le i \neq j \le k} (z_i/z_j;q)_\infty 
    = \mathsf{W}_k(z_1,\ldots,z_k) \prod_{1 \le i < j \le k} (z_j^{-1} - z_i^{-1}) \, . \label{eq:root_prod}
\end{align}
Then, we obtain a determinantal form (i.e., a factorized form) of the partition function,
\begin{align}
    Z_{X,\mathsf{I}}^\theta & = \Theta(a_\mathsf{I}) \prod_{i=1}^k \frac{x^{-\log_q a_i} \theta(a_i)^\kappa }{(q;q)_\infty \prod_{\beta (\neq I(i))}^n (a_i/a_\beta;q)_\infty} \nonumber \\
    & \qquad \times \sum_{d_1,\ldots,d_k \ge 0} x^{|d|} \prod_{i=1}^k \frac{q^{-\kappa_\text{eff}\binom{d_i+1}{2}} (-a_i)^{\kappa_\text{eff}d_i}}{\prod_{\beta=1}^n (q^{-1} a_i/a_\beta;q^{-1})_{d_i}} \prod_{1 \le i < j \le k} (a_j^{-1} q^{d_j} - a_i^{-1} q^{d_i}) \nonumber \\
    & = \Theta(a_\mathsf{I}) \det_{1 \le i, j \le k} \phi_{\mathsf{I}(i)}(xq^{j-1}) \, , \label{eq:Gr_det}
\end{align}
where we define the effective CS level $\kappa_\text{eff} = \kappa + k$ and
\begin{align}
    \phi_\alpha(x) = \frac{x^{-\log_q a_\alpha} \theta(a_\alpha)^\kappa}{(q;q)_\infty \prod_{\beta (\neq \alpha)}^n (a_\alpha/a_\beta;q)_\infty} \sum_{d \ge 0} x^d \frac{q^{-\kappa_\text{eff}\binom{d+1}{2}} (-a_\alpha)^{\kappa_\text{eff}d}}{\prod_{\beta=1}^n(q^{-1}a_\alpha/a_\beta;q^{-1})_d} \, , \qquad \alpha = 1, \ldots, n \, .
\end{align}
See \cite[Section 4]{Kimura:2026} for a related discussion.
This $q$-Casoratian minor formula indeed proves the K-theoretic version of the HV formula~\cite{Hori:2000kt},
\begin{align}
    \frac{Z_{X,\mathsf{I}}^\theta}{\Theta(a_\mathsf{I})} = \left. \Delta_q \prod_{i=1}^k \phi_{\mathsf{I}(i)} (x_i) \right|_{x_i = x} \, , \qquad \Delta_q = \prod_{1 \le i < j \le k} (q^{D_j} - q^{D_i}) \, , \qquad D_i = x_i \frac{\partial}{\partial x_i} \, . \label{eq:HV_formula}
\end{align}
Compared with the previous case \eqref{eq:q-hypergeometric_Pn}, we have the CS level shift $\kappa \mapsto \kappa_\text{eff}$ for the vortex partition function. 
Therefore, $\{\phi_\alpha\}_{\alpha=1,\ldots,n}$ are independent solutions to the following $q$-difference equation,
\begin{align}
    \left[ q^{\kappa_\text{eff} D} \prod_{\beta=1}^n (1 - q^{-D}/a_\beta) - (-1)^{\kappa_\text{eff}} x \right] y(x) = 0 \, . \label{eq:q-diff_Gr}
\end{align}
More concretely, we have 
\begin{align}
    \sum_{r = 0}^n (-1)^r e_r(a^{-1}) y(x q^{-r+\kappa_\text{eff}}) - (-1)^{\kappa_\text{eff}} x y(x) = 0 \, , \label{eq:q-diff_Gr_exp}
\end{align}
where $e_r(a^{-1}) = e_r(a_1^{-1},\ldots,a_n^{-1})$ is the $r$-th elementary symmetric polynomial.
In particular, for $\kappa_\text{eff} = 0$, we may rewrite 
\begin{align}
    y(x) = \sum_{m \ge 0} x^m \left[\sum_{r=0}^{n-1} (-1)^{n-r-1} e_{n-r}(a^{-1}) y(xq^{r-n}) \right] \, . \label{eq:q-diff_Gr_exp2}
\end{align}

We consider a $q$-difference equation for the $q$-Casoratian minor.
Set $\kappa_\text{eff} = 0$, and define a $q$-difference operator of order $n$,
\begin{align}
    \widehat{\mathscr{A}}_i = \prod_{\beta=1}^n (1 - q^{-D_i}/a_\beta) - x_i \, ,
\end{align}
such that $\widehat{\mathscr{A}}_i \phi_\alpha(x_i) = 0$ for any $\alpha = 1, \ldots,n$.
A $q$-difference operator $\widehat{\mathscr{A}}_\mathsf{I}$ annihilating $\widetilde{Z}_{\theta,\mathsf{I}} = \Theta(a_\mathsf{I}) \Delta_q \prod_{i=1}^{k} \phi_{\mathsf{I}(i)} (x_i)$, is given by a similarity transform,
\begin{align}
    \widehat{\mathscr{A}}_\mathsf{I} = \Delta_q \left( \prod_{i=1}^k \widehat{\mathscr{A}}_{\mathsf{I}(i)} \right) \Delta_q^{-1}
\end{align}
since
\begin{align}
    \frac{\widehat{\mathscr{A}}_\mathsf{I} \widetilde{Z}_{\theta,\mathsf{I}}}{\Theta(a_\mathsf{I})} = \widehat{\mathscr{A}}_\mathsf{I} \Delta_q \prod_{i=1}^k \phi_{\mathsf{I}(i)}(x_i) = \Delta_q \prod_{i=1}^k \widehat{\mathscr{A}}_i \phi_{\mathsf{I}(i)}(x_i) = 0 \, .
\end{align}
Then, we take the limit, $\lim_{x_i \to x} \widehat{\mathscr{A}}_\mathsf{I}$, to obtain a $q$-difference operator, which annihilates the $q$-Casoratian minor.
As $\widehat{\mathscr{A}}_i$ is of order $n$, the $q$-Casoratian minor of size $k$ obeys an order $\binom{n}{k}$ $q$-difference equation.

This leads to the quantum $\mathscr{D}_q$-module structure of the grassmannian.
The HV formula implies that the grassmannian $q$-difference module is obtained from the $k$-fold exterior power of the projective space $q$-difference module, $\mathscr{M}_{\operatorname{Gr}(k,n)} = \Lambda^k \mathscr{M}_{\mathbb{P}^{n-1}}$, and at the level of solution sheaves, this gives $\operatorname{Sol}(\mathscr{M}_{\operatorname{Gr}(k,n)}) = \Lambda^k \operatorname{Sol}(\mathscr{M}_{\mathbb{P}^{n-1}})$, which is an analog of the cohomology relation, $H^\bullet(\operatorname{Gr}(k,n)) \cong \Lambda^k H^\bullet(\mathbb{P}^{n-1})$.
More precisely, we first consider the external tensor product of $k$ copies of the projective space $q$-difference module, $\mathscr{M}_{\mathbb{P}^{n-1}}^{\boxtimes k} = \mathscr{M}_{\mathbb{P}^{n-1}} \boxtimes \cdots \boxtimes \mathscr{M}_{\mathbb{P}^{n-1}}$, which is a $q$-difference module over the $k$-dimensional torus $(x_1,\ldots,x_k) \in (\mathbb{C}^\times)^k$.
Then, we apply the anti-symmetrization operator $\Delta_q$, and restrict it to the diagonal subtorus, $x_1 = \cdots = x_k$.
The resulting $q$-difference module has rank $\binom{n}{k}$ for $\kappa_\text{eff} = 0$, which agrees with the number of Higgs vacua and with the rank of $K(\operatorname{Gr}(k,n))$.

\subsubsection{Factorization of the kernel}\label{sec:factorization_kernel}

The factorization property can be seen at the level of the contour integral. 
We consider the integral kernel for the abelian quotient with $k$ distinct FI parameters,
\begin{align}
    \widetilde{\omega}_X^\theta(x_1,\ldots,x_k;z_1,\ldots,z_k) = \theta(\det z) \prod_{1 \le i \neq j \le k} \left(\frac{z_i}{z_j};q\right)_\infty \prod_{i=1}^k \frac{x_i^{-\log_q z_i} \theta(z_i)^\kappa}{\prod_{\beta=1}^n (z_i/a_\beta;q)_\infty} \frac{\dd{z_i}}{2\pi\ii z_i} \, . \label{eq:kernel_abelian}
\end{align}
It clearly reproduces the original integral kernel~\eqref{eq:omega_Gr_theta} in the limit $x_i \to x$ for all $i = 1,\ldots,k$.
By the identity~\eqref{eq:root_prod}, we have
\begin{align}
    \widetilde{\omega}_X^\theta(x_1,\ldots,x_k;z_1,\ldots,z_k) & = \Theta(z) \prod_{1 \le i < j \le k} (z_j^{-1} - z_i^{-1}) \prod_{i=1}^k \frac{x_i^{-\log_q z_i} \theta(z_i)^\kappa}{\prod_{\beta=1}^n (z_i/a_\beta;q)_\infty} \frac{\dd{z_i}}{2\pi\ii z_i} \nonumber \\
    & = \Theta(z) \Delta_q \prod_{i=1}^k \omega_{\mathbb{P}^{n-1}}(x_i,z_i) \, ,
\end{align}
which is the HV formula at the level of the integral kernel.
By the relation~\eqref{eq:q-const}, the contribution of the elliptic functions turns out to be quasi periodic under $q$-shift.
Evaluating the integral associated with the contour $\mathcal{C}_I$, we obtain
\begin{align}
    \widetilde{Z}_{X,\mathsf{I}}^\theta = \int_{\mathcal{C}_\mathsf{I}} \widetilde{\omega}_X^\theta & = \Theta(a_\mathsf{I}) \Delta_q \prod_{i=1}^{k} \phi_{\mathsf{I}(i)} (x_i) \xrightarrow{x_i \to x} Z_{X,\mathsf{I}}^\theta \, .
\end{align}
Here we compute
\begin{align}
    \frac{{\omega}_X^\theta(x;\{z_i q^{-d_i}\}_{i=1,\ldots,k})}{{\omega}_X^\theta(x;z_1,\ldots,z_k)} = x^{|d|} \prod_{1 \le i < j \le k} \frac{z_j^{-1} q^{d_j} - z_i^{-1} q^{d_i}}{z_j^{-1} - z_i^{-1}} \prod_{i=1}^k \frac{q^{-\kappa_\text{eff}\binom{d_i+1}{2}} (-z_i)^{\kappa_\text{eff}d_i}}{\prod_{\beta=1}^n (q^{-1}z_i/a_\beta;q^{-1})_{d_i}} \, .
\end{align}
This expression agrees with the degree $d = (d_1,\ldots,d_k)$ contribution of the I-function under the map, $z_i \mapsto P_i$, which is interpreted as the K-theoretic Kirwan map together with an appropriate completion.

From the simplest case $|d| = 1$, we define the $q$-connection,
\begin{align}
    \Omega_{X,i}(z_1,\ldots,z_k) = x (-z_i)^{\kappa_\text{eff}} \prod_{j (\neq i)}^k \frac{z_i - qz_j}{z_i - z_j}  \prod_{\beta=1}^n \left( 1 - \frac{z_i}{a_\beta} \right)^{-1} \, , \qquad i = 1,\ldots,k \, ,
\end{align}
and, for a $k$-variable function $f$, we define a twisted $q$-difference operator,
\begin{align}
    \nabla_{X,i} f(z_1,\ldots,z_k) = \Omega_{X,i}(z_1,\ldots,z_k)  f(z_1,\ldots,z_i q^{-1},\ldots,z_k) - f(z_1,\ldots,z_k) \, .
\end{align}
This $q$-connection is indeed flat,
\begin{align}
     \Omega_{X,i}(z_1,\ldots,z_k) \Omega_{X,j}(z_1,\ldots,z_i q^{-1},\ldots,z_k) = \Omega_{X,j}(z_1,\ldots,z_k) \Omega_{X,i}(z_1,\ldots,z_j q^{-1},\ldots,z_k) \, .
\end{align}
Then, in the associated twisted $q$-de Rham cohomology, two insertions are identified whenever they differ by a $\nabla_X$-exact term:
\begin{align}
    g \sim g + \nabla_{X,i} f \, , \label{eq:double_quantum_K-ring}
\end{align}
This relation reduces to the quantum equivariant K-ring relations in the limit $q \to 1$, $\left. \Omega_{X,i}\right|_{q\to1} = 1$, which is equivalent to
\begin{align}
    x (-z_i)^{\kappa_\text{eff}}  = \prod_{\beta=1}^n \left( 1 - \frac{z_i}{a_\beta} \right) \, , \qquad i = 1,\ldots,k \, . \label{eq:quantum_K-ring}
\end{align}
We will revisit this relation in \autoref{sec:Gr_semiclassical}.

\subsubsection{Wilson loop insertion}\label{sec:Wilson_loop}

The Coulomb branch localization formula is also available in the presence of the Wilson loop operators along $\mathbb{S}^1$.
Let $s_\lambda$ be the Schur polynomial~\eqref{eq:Schur_def}, which is a character polynomial of $\mathrm{U}(k)$; $\lambda$ parametrizes the irreducible representation of $\mathrm{U}(k)$.
We write $\chi_\lambda = s_\lambda(z_1^{-1},\ldots,z_k^{-1})$.
Then, the Coulomb branch formula in the presence of the Wilson loops is given by
\begin{align}
    Z_{X}^\theta(\lambda_1 ,\ldots, \lambda_L) = \frac{1}{k!} \int \chi_{\lambda_1} \cdots \chi_{\lambda_L} \omega_X^\theta \, .
\end{align}
The integral can be efficiently evaluated via the abelian quotient.
We write $\widetilde{Z}_{\theta}(\varnothing) = \widetilde{Z}_{\theta}$.
Then, we have
\begin{align}
    \widetilde{Z}_X^{\theta} (\lambda_1 , \ldots, \lambda_L)
    = \int \chi_{\lambda_1} \cdots \chi_{\lambda_L} \widetilde{\omega}_X^\theta 
    = \left( \prod_{l=1}^L s_{\lambda_l}(q^{D_1},\ldots,q^{D_k}) \right) \widetilde{Z}_X^{\theta}(\varnothing) \, ,
\end{align}
hence the Wilson loops are generated by Schur polynomials of $q$-difference operators, which reproduces the discussion in~\cite[\S3]{Jockers:2019lwe}.

\paragraph{One-point function}
Let us focus on the simplest case $L=1$.
For a partition $\lambda$ of length $k$, we denote the inverse one by $\lambda^\vee = (\lambda^\vee_j)_{j=1,\ldots,k}$ with $\lambda^\vee_j = \lambda_{k+1-j}$.
We compute the one-point function of the Wilson loop operator as follows,
\begin{align}
    \widetilde{Z}_X^{\theta}(\lambda) 
    = \int \chi_{\lambda} \widetilde{\omega}_X^\theta 
    = \det_{1 \le i, j \le k} q^{(\lambda_j^\vee + j - 1)D_i} \int \Theta(z) \prod_{i=1}^k \omega_{\mathbb{P}^{n-1}}(x_i;z_i) \frac{\dd{z_i}}{2\pi\ii z_i} \, . 
\end{align}
The integral associated with the contour $\mathcal{C}_\mathsf{I}$ is thus given by
\begin{align}
    \widetilde{Z}_{X,\mathsf{I}}^\theta(\lambda) = \Theta(a_\mathsf{I}) \det_{1 \le i, j \le k} q^{(\lambda_j^\vee + j - 1)D_i} \prod_{i=1}^{k} \phi_{\mathsf{I}(i)} (x_i) \xrightarrow{x_i = x} Z_{X,\mathsf{I}}^\theta(\lambda) = \Theta(a_\mathsf{I}) \det_{1 \le i, j \le k} \phi_{\mathsf{I}(i)}(x q^{\lambda_j^\vee + j - 1}) \, . \label{eq:WL_det}
\end{align}
From this expression, we see the action of $q^D$ on the Wilson loop partition function as follows,
\begin{align}
    (q^D)^l Z_{X,\mathsf{I}}^\theta(\lambda) = Z_{X,\mathsf{I}}^\theta(\lambda+l) \, , \qquad 
    \lambda+l = (\lambda_1+l,\ldots,\lambda_k+l) \, .
\end{align}

The determinantal formula motivates us to define a quantum version of the Schur function, 
\begin{align}
    \mathscr{S}_\lambda(a_\mathsf{I}^{-1}) = \frac{Z_{X,\mathsf{I}}^\theta(\lambda)}{Z_{X,\mathsf{I}}^\theta(\varnothing)} = \frac{\det_{1 \le i, j \le k} \phi_{\mathsf{I}(i)}(x q^{\lambda_j^\vee + j - 1})}{\det_{1 \le i, j \le k} \phi_{\mathsf{I}(i)}(x q^{j - 1})} \, ,     
\end{align}
which reproduces the classical Schur polynomial in the limit $x \to 0$,
\begin{align}
    \lim_{x \to 0} \mathscr{S}_\lambda(a_\mathsf{I}^{-1}) = s_\lambda(a_\mathsf{I}^{-1}) \, .
\end{align}
Set $\kappa_\text{eff} = 0$ for simplicity. 
As seen from the $q$-difference equation~\eqref{eq:q-diff_Gr_exp2}, $\phi_\alpha(xq^{n})$ can be expressed as a linear combination of $\{\phi_\alpha(xq^{r})\}_{r=0,\ldots,n-1}$.
Hence, the one-point function for $\lambda \not\in (n-k)^k$ with $\ell(\lambda) \le k$ can be expressed as a linear combination of those for $\lambda \in (n-k)^k$ with $\# \{ \lambda \in (n-k)^k \} = \binom{n}{k}$, which agrees with the number of the Higgs vacua.
At the level of the integral kernel, we interpret $\chi_\lambda$ as a twisted $q$-de Rham cohomology element, for which we impose the $q$-deformed quantum equivariant K-ring relations~\eqref{eq:double_quantum_K-ring}, and the integral $Z_{X,\mathsf{I}}^\theta(\lambda)$ is the corresponding $q$-period evaluated with the contour $\mathcal{C}_I$.

\paragraph{Two-point function}

Let us then consider the case $L = 2$, which plays a fundamental role in the algebraic structure of the Wilson loop operators.
First of all, the character polynomials (i.e., Schur polynomials) obey the Littlewood--Richardson (LR) fusion rule,
\begin{align}
    \chi_\lambda \chi_\mu = \sum_\nu N_{\lambda\mu}^\nu \chi_\nu \, ,
\end{align}
where $N_{\lambda\mu}^\nu \in \mathbb{Z}_{\ge 0}$ is called the LR coefficient.
The character polynomial is identically zero when the length of partition is strictly larger than the number of variables, $\chi_\lambda = 0$ for $\ell(\lambda) > k$, while the LR coefficient does not depend on $k$.
We insert the product of two character polynomials to compute the two-point function of the Wilson loops,
\begin{align}
    Z_{X,\mathsf{I}}^\theta(\lambda,\mu) = \sum_{\nu} N_{\lambda\mu}^\nu Z_{X,\mathsf{I}}^\theta(\nu) \, .
\end{align}
The multiple insertion is rather a convolution, so that $Z_{X,\mathsf{I}}^\theta(\lambda, \mu) \neq Z_{X,\mathsf{I}}^\theta(\lambda) Z_{X,\mathsf{I}}^\theta(\mu)$.

As discussed above, the one-point function of the Wilson loop for $\lambda \not \in (n-k)^k$ is given by a linear combination of those for $\nu \in (n-k)^k$.
Hence, we obtain the modified fusion rule for the Wilson loop operators~\cite{Jockers:2019lwe},
\begin{align}
    Z_{X,\mathsf{I}}^\theta(\lambda,\mu) = \sum_{\nu \in (n-k)^k} N_{\lambda\mu}^\nu(x,a) Z_{X,\mathsf{I}}^\theta(\nu) \, , \qquad N_{\lambda\mu}^\nu(x,a) \in R(T) \llbracket x \rrbracket \, ,
\end{align}
where we denote by $R(T)$ the representation ring of the torus $T = (\mathbb{C}^\times)^n \ni (a_1,\ldots,a_n)$.
The LR coefficients depend on $(x,a)$, but not on $q$, since the coefficients of the $q$-difference equation \eqref{eq:q-diff_Gr_exp2} depend only on $(x,a)$.

\subsection{$T^*\!\operatorname{Gr}(k,n)$}\label{sec:TGr}

Let us explore an extension of the grassmannian GLSM, $\mathrm{U}(k)$ gauge theory with $n$ fundamental, $n$ anti-fundamental, and one adjoint chiral multiplets, which is a GLSM describing the cotangent bundle of the grassmannian $T^*\!\operatorname{Gr}(k,n)$.
We assign the Neumann boundary condition for the fundamental and adjoint chiral multiplets, while the Dirichlet boundary condition is assigned for the anti-fundamental chiral multiplets.
Then, the fundamental and the anti-fundamental chiral multiplets form the $\mathcal{N} = 4$ hypermultiplet, and the $\mathcal{N}=2$ vector multiplet and the adjoint chiral multiplet form the $\mathcal{N} = 4$ vector multiplet.
We turn off the CS term $\kappa = 0$ to obtain a mass-deformed $\mathcal{N} = 4$ gauge theory in three dimensions.
The Coulomb branch localization formula is given by
\begin{align}
    Z_X = \frac{1}{k!} \int \prod_{1 \le i \neq j \le k} \frac{(z_i/z_j;q)_\infty}{(tz_i/z_j;q)_\infty} \prod_{\substack{i=1,\ldots,k \\ \alpha = 1, \ldots, n}} \frac{(z_i/b_\alpha;q)_\infty}{(z_i/a_\alpha;q)_\infty} \prod_{i=1}^k x^{-\log_q z_i} \frac{\dd{z_i}}{2\pi\ii z_i} \, . \label{eq:N=4_contour}
\end{align}
Compared to the previous case~\eqref{eq:Z_Gr}, we have additional flavor fugacities for anti-fundamental and adjoint matters, $b = (b_\alpha)_{\alpha = 1, \ldots,n}$ and $t$, respectively, where we need to set $b_\alpha = a_\alpha / t$ to be consistent with $\mathcal{N} = 4$ R-symmetry.
Nevertheless, we keep using generic $b = (b_\alpha)_{\alpha = 1, \ldots,n}$ in the following analysis.

We may apply the same residue calculus as before,\footnote{%
In general, there are additional poles from $\prod_{i \neq j}^k (t z_i/z_j;q)_\infty$ in the denominator, which are located at $z_i = t^{-1} q^{-d} z_j$ for $d \in \mathbb{Z}_{\ge 0}$. The corresponding residues will be zero when $b_\alpha = a_\alpha / t$.}
\begin{align}
    Z_{X,\mathsf{I}} & = \prod_{1 \le i \neq j \le k} \frac{(a_i/a_j;q)_\infty}{(ta_i/a_j;q)_\infty} \prod_{i=1}^k \frac{x^{-\log_q a_i} \prod_{\beta=1}^n (a_i/b_\beta;q)_\infty}{(q;q)_\infty \prod_{\beta (\neq \mathsf{I}(i))}^n (a_i/a_\beta;q)_\infty} \nonumber \\    
    & \quad \times \sum_{d_1,\ldots,d_k \ge 0} x^{|d|} \prod_{1 \le i \neq j \le k} \frac{\left( q^{-d_i+d_j} a_i/a_j;q\right)_{\infty}}{\left(t q^{-d_i+d_j} a_i/a_j;q\right)_{\infty}} \frac{\left( t a_i/a_j;q\right)_{\infty}}{\left(a_i/a_j;q\right)_{\infty}} \prod_{i=1}^k \prod_{\beta=1}^n \frac{(q^{-1} a_i/b_\beta;q^{-1})_{d_i}}{(q^{-1} a_i / a_\beta;q^{-1})_{d_i}} \, ,
\end{align}
from which we identify the vortex partition function,
\begin{align}
    Z^\text{vor}_{X,\mathsf{I}} & =  \sum_{d_1,\ldots,d_k \ge 0} x^{|d|} \prod_{1 \le i \neq j \le k} \frac{\left( q^{-d_i+d_j} a_i/a_j;q\right)_{\infty}}{\left(t q^{-d_i+d_j} a_i/a_j;q\right)_{\infty}} \frac{\left( t a_i/a_j;q\right)_{\infty}}{\left(a_i/a_j;q\right)_{\infty}} \prod_{i=1}^k \prod_{\beta=1}^n \frac{(q^{-1} a_i/b_\beta;q^{-1})_{d_i}}{(q^{-1} a_i / a_\beta;q^{-1})_{d_i}} \nonumber \\
    & = \sum_{d_1,\ldots,d_k \ge 0} x^{|d|} \prod_{1 \le i \neq j \le k} \prod_{\ell = - \infty}^{d_i - d_j} \frac{1 - q^{-\ell} a_i/a_j}{1 - t q^{-\ell} a_i/a_j} \prod_{\ell = - \infty}^{0} \frac{1 - t q^{-\ell} a_i/a_j}{1 - q^{-\ell} a_i/a_j} \prod_{i=1}^k \prod_{\beta=1}^n \prod_{\ell = 1}^{d_i} \frac{1-q^{-\ell} a_i/b_\beta}{1-q^{-\ell} a_i/a_\beta} \, .
\end{align}
Then, $(1-q)Z^\text{vor}_{X,\mathsf{I}}$ agrees with the K-theoretic I-function for the odd tangent bundle $X = \Pi T \! \operatorname{Gr}(k,n)$%
\footnote{%
The odd bundle is also denoted by $T \! \operatorname{Gr}(k,n)[1]$.
See \autoref{sec:I-func}.
}
(balanced I-function) after the localization and the replacement $q \mapsto q^{-1}$ together with the specialization $b_\alpha = a_\alpha / t$~\eqref{eq:I-func_TGr}, which differs from that for the cotangent bundle $T^*\!\operatorname{Gr}(k,n)$ as emphasized in~\cite[\S5]{Givental:2021dbg}.

In order to obtain the I-function for $X = T^*\!\operatorname{Gr}(k,n)$, we still consider the same bulk theory, but assign different boundary conditions, the Neumann boundary condition for the fundamental and the anti-fundamental chiral multiplets, and the Dirichlet boundary condition for the adjoint chiral multiplet.
Then, we instead obtain the following integral formula,  
\begin{align}
    Z_{X} = \frac{1}{k!} \int \prod_{1 \le i \neq j \le k} \left(\frac{z_i}{z_j},\frac{q}{t}\frac{z_i}{z_j};q\right)_\infty \prod_{i=1}^k \frac{x^{-\log_q z_i}}{\prod_{\alpha = 1}^n (z_i/a_\alpha,q b_\alpha/z_i;q)_\infty} \frac{\dd{z_i}}{2\pi\ii z_i} \, . 
\end{align}
Comparing the integral kernel \eqref{eq:N=4_contour}, they differ by the theta functions, which are interpreted as the elliptic genus of the boundary $\mathcal{N}=(0,2)$ theory,
\begin{align}
    \omega_{T^*\!\operatorname{Gr}}(x;z_1,\ldots,z_k) = \omega_{\Pi T\!\operatorname{Gr}}(x;z_1,\ldots,z_k)  \prod_{1 \le i \neq j \le k} \theta\left(\frac{q}{t}\frac{z_i}{z_j}\right) \prod_{i=1}^k \prod_{\alpha = 1}^n \theta(qb_\alpha/z_i)^{-1} \, .
\end{align}
Applying the same residue analysis, we compute the partition function, 
\begin{align}
    Z_{X,\mathsf{I}} & = \prod_{1 \le i \neq j \le k} \left(\frac{a_i}{a_j},\frac{q}{t}\frac{a_i}{a_j};q\right)_\infty \prod_{i=1}^k \frac{x^{-\log_q a_i}}{(q;q)_\infty \prod_{\beta (\neq \mathsf{I}(i))}^n (a_i/a_\beta;q)_\infty \prod_{\beta=1}^n (qb_\beta/a_i;q)_\infty} \nonumber \\    
    & \quad \times \sum_{d_1,\ldots,d_k \ge 0} x^{|d|} \prod_{1 \le i \neq j \le k} \frac{\left(q^{-d_i+d_j} a_i/a_j;q\right)_{\infty}}{\left(a_i/a_j;q\right)_{\infty}} \frac{\left(t^{-1} q^{-d_i+d_j+1} a_i/a_j;q\right)_{\infty}}{\left( t^{-1} q a_i/a_j;q\right)_{\infty}} \prod_{i=1}^k \prod_{\beta=1}^n \frac{(q b_\beta / a_i;q)_{d_i}}{(q^{-1} a_i / a_\beta;q^{-1})_{d_i}} \, ,
\end{align}
from which we identify the vortex partition function, 
\begin{align}
    Z^\text{vor}_{X,\mathsf{I}} & =  \sum_{d_1,\ldots,d_k \ge 0} x^{|d|} \prod_{1 \le i \neq j \le k} \frac{\left(q^{-d_i+d_j} a_i/a_j;q\right)_{\infty}}{\left(a_i/a_j;q\right)_{\infty}} \frac{\left(t^{-1} q^{-d_i+d_j+1} a_i/a_j;q\right)_{\infty}}{\left(t^{-1} q a_i/a_j;q\right)_{\infty}} \prod_{i=1}^k \prod_{\beta=1}^n \frac{(q b_\beta / a_i;q)_{d_i}}{(q^{-1} a_i / a_\beta;q^{-1})_{d_i}} \nonumber \\
    & = \sum_{d_1,\ldots,d_k \ge 0} x^{|d|} \prod_{1 \le i \neq j \le k} \frac{\prod_{\ell = -\infty}^{d_i-d_j} (1 - q^{-\ell} a_i/a_j) (1 - t^{-1} q^{-\ell+1} a_i/a_j)}{\prod_{\ell = -\infty}^0 (1 - q^{-\ell} a_i/a_j) (1 - t^{-1} q^{-\ell+1} a_i/a_j)} \prod_{i=1}^k \prod_{\beta=1}^n \prod_{\ell = 1}^{d_i} \frac{1-q^\ell b_\beta/a_i}{1-q^{-\ell} a_i/a_\beta} \, ,
\end{align}
where $(1-q)Z^\text{vor}_{X,\mathsf{I}}$ reproduces the K-theoretic I-function for the cotangent bundle $X = T^*\!\operatorname{Gr}(k,n)$ under the identification of the cotangent weight $v = q/t$~\eqref{eq:I-func_T*Gr}.
 
\paragraph{Integral kernel}

Let us discuss the twisted $q$-de Rham cohomology for $X = \Pi T\!\operatorname{Gr}(k,n)$.
The integral kernel $\omega_X(x;z_1,\ldots,z_k)=\omega_{\Pi T\!\operatorname{Gr}}(x;z_1,\ldots,z_k)$ is not simply factorized as in the case of $\operatorname{Gr}(k,n)$.
Nevertheless, we have
\begin{align}
    \frac{\omega_X(x;\{z_i q^{-d_i}\}_{i=1,\ldots,k})}{\omega_X(x;\{z_i \}_{i=1,\ldots,k}\})} = x^{|d|} \prod_{1 \le i \neq j \le k} \frac{\left( q^{-d_i+d_j} z_i/z_j;q\right)_{\infty}}{\left(t q^{-d_i+d_j} z_i/z_j;q\right)_{\infty}} \frac{\left( t z_i/z_j;q\right)_{\infty}}{\left(z_i/z_j;q\right)_{\infty}} \prod_{i=1}^k \prod_{\beta=1}^n \frac{(q^{-1} z_i/b_\beta;q^{-1})_{d_i}}{(q^{-1} z_i / a_\beta;q^{-1})_{d_i}} \, ,
\end{align}
and the simplest case $|d| = 1$ gives
\begin{align}
    \Omega_{X,i}(z_1,\ldots,z_k) = x \prod_{j (\neq i)}^k \frac{(z_i - q z_j)(z_i - t z_j)}{(z_i - z_j)(t z_i - q z_j)} \prod_{\beta=1}^n \frac{1 - z_i/b_\beta}{1 - z_i/a_\beta} \, , \qquad i = 1,\ldots,k \, .
\end{align}
Then, we define a twisted $q$-difference operator,
\begin{align}
    \nabla_{X,i} f(z_1,\ldots,z_k) = \Omega_{X,i}(z_1,\ldots,z_k)  f(z_1,\ldots,z_i q^{-1},\ldots,z_k) - f(z_1,\ldots,z_k) \, ,
\end{align}
and we have an equivalence relation for two insertions, $g \sim g + \nabla_{X,i} f$.
This relation is indeed related to the vortex analog of the $qq$-character~\cite{Nekrasov:2015wsu,Kimura:2015rgi} \cite{Haouzi:2020bso}.
In the limit $q \to 1$, it reduces to 
\begin{align}
    x \prod_{j (\neq i)}^k \frac{z_i - t z_j}{t z_i - z_j} =  \prod_{\beta=1}^n \frac{1 - z_i/a_\beta}{1 - z_i/b_\beta} \, , \qquad i = 1, \ldots, k \, , \label{eq:quantum_K-ring_cotangent}
\end{align}
which is identified with the Bethe Ansatz Equation as seen in \autoref{sec:Gr_semiclassical}.

\subsection{Semiclassical analysis}\label{sec:Gr_semiclassical}

We study the semiclassical behavior of the partition function in the limit $q \to 1$. 
We put $q = \ee^{-\epsilon}$ with $\operatorname{Re} \epsilon > 0$.
The $q$-shifted factorial behaves as follows in the regime $\epsilon \ll 1$,
\begin{align}
    (z;q)_\infty = \exp \left(\frac{1}{\epsilon} \left( - \operatorname{Li}_2(z) + O(\epsilon) \right) \right) \, , \qquad \operatorname{Li}_2(z) = \sum_{n = 1}^\infty \frac{z^n}{n^2} \, .
\end{align}
The dilogarithm obeys the reflection relation,
\begin{align}
    \operatorname{Li}_2(z) + \operatorname{Li}_2(z^{-1}) = - \frac{1}{2} \log^2 (-z) - \frac{\pi^2}{6} \, ,
\end{align}
and thus the asymptotic behavior of the theta function is given as follows, 
\begin{align}
    \theta(z;q) & = \exp \left( - \frac{1}{\epsilon} \left( \operatorname{Li}_2(z) + \operatorname{Li}_2(z^{-1}) + O(\epsilon) \right) \right) %\nonumber \\ &
    = \exp \left( \frac{1}{\epsilon} \left(  \frac{1}{2} \log^2 (-z) + \frac{\pi^2}{6} + O(\epsilon) \right) \right) \, .
\end{align}

\subsubsection{$\operatorname{Gr}(k,n)$}

In the limit $q \to 1$ ($\epsilon \to 0$), the partition function for $X = \operatorname{Gr}(k,n)$ shown in \eqref{eq:Z_Gr} behaves as follows,
\begin{align}
    Z_X = \frac{1}{k!} \int \exp \left( \frac{1}{\epsilon} \left(\widetilde{W}_X(z) + O(\epsilon) \right) \right) \prod_{i=1}^k \frac{\dd{z_i}}{2\pi\ii z_i} \, ,
\end{align}
where $\widetilde{W}_X$ is the effective twisted superpotential given by
\begin{align}
    \widetilde{W}_X(z) = \log x \sum_{i=1}^k \log z_i + \frac{\kappa}{2} \sum_{i=1}^k \log^2 (-z_i) - \sum_{1 \le i \neq j \le k}  \operatorname{Li}_2(z_i/z_j) + \sum_{\alpha=1}^n \sum_{i=1}^k \operatorname{Li}_2(z_i/a_\alpha) + \text{const.} \, .
\end{align}
The twisted F-term equations are the critical point equation with respect to the effective twisted superpotential,
\begin{align}
    \exp \left( z_i \frac{\partial}{\partial z_i} \widetilde{W}_X \right) = 1 \, , \qquad i = 1, \ldots, k \, ,
\end{align}
which are equivalent to
\begin{align}
    - x \frac{(-z_i)^{\kappa_\text{eff}}}{z_1 \cdots z_k} = \prod_{\alpha = 1}^n \left( 1 - \frac{z_i}{a_\alpha} \right) \, , \qquad i = 1, \ldots, k \, .
\end{align}

We then include the CS term of the determinant representation~\eqref{eq:Z_Gr_theta}.
The twisted superpotential is modified as follows,
\begin{align}
    \widetilde{W}_{X}^\theta(z) = \widetilde{W}_{X}(z) + \frac{1}{2} \log^2 (-z_1 \cdots z_k) + \text{const.} \, ,
\end{align}
and the corresponding twisted F-term equations are given by,
\begin{align}
    x (-z_i)^{\kappa_\text{eff}} = \prod_{\alpha = 1}^n \left( 1 - \frac{z_i}{a_\alpha} \right) \, , \qquad i = 1, \ldots, k \, , \label{eq:quantum_K-ring2}
\end{align}
which reproduce the previous result~\eqref{eq:quantum_K-ring}.
Expanding the RHS of these equations, we have
\begin{align}
    x (-z_i)^{\kappa_\text{eff}} = 1 + \cdots + \frac{z_i^{n}}{a_1 \cdots a_n} \, .
\end{align}
Hence, in order to obtain a degree $n$ polynomial relation, we have the CS level window, $0 \le \kappa_\text{eff} \le n \iff -k \le \kappa \le n-k$.
In the absence of the theta term $\theta(\det z)$, this window is shifted by one.
When $\kappa_\text{eff}=0$, the twisted F-term equations \eqref{eq:quantum_K-ring2} agree with the equivariant quantum K-ring relations for $X = \operatorname{Gr}(k,n)$,
\begin{align}
    x = \prod_{\alpha = 1}^n \left( 1 - \frac{z_i}{a_\alpha} \right) \, , \qquad i = 1, \ldots, k \, . \label{eq:quantum_K-ring3}
\end{align}
Hence, in the semiclassical regime, the partition function is evaluated by summing over the vacua characterized by these relations as in the case of the topologically twisted indices~\cite{Benini:2015noa,Closset:2016arn}.

\subsubsection{$\Pi T\!\operatorname{Gr}(k,n)$}

For $X = \Pi T\!\operatorname{Gr}(k,n)$, we obtain the effective twisted superpotential in the limit $q \to 1$ from the partition function~\eqref{eq:N=4_contour} as follows,
\begin{align}
    \widetilde{W}_X(z) = \log x \sum_{i=1}^k \log z_i - \sum_{1 \le i \neq j \le k} \left( \operatorname{Li}_2(z_i/z_j) - \operatorname{Li}_2(t z_i/z_j) \right) + \sum_{\alpha=1}^n \sum_{i=1}^k \left( \operatorname{Li}_2(z_i/a_\alpha) - \operatorname{Li}_2(z_i/b_\alpha) \right) \, ,
\end{align}
and the corresponding twisted F-term equations are given by
\begin{align}
    x \prod_{j (\neq i)}^k \frac{z_i - t z_j}{t z_i - z_j} = \prod_{\alpha=1}^n \frac{1 - z_i/a_\alpha}{1 - z_i/b_\alpha} \, , \qquad i = 1,\ldots, k \, ,
\end{align}
which reproduce the previous result~\eqref{eq:quantum_K-ring_cotangent}.
We parametrize $z_i = \exp(\beta \xi_i)$, $a_\alpha = \exp(\beta (\lambda_\alpha+s_\alpha))$, $b_\alpha = \exp(\beta (\lambda_\alpha-s_\alpha))$, $t = \exp(\beta \hbar)$, and we define $[z] = 2 \sinh (\beta z /2)$, where $\beta$ is the circumference of the compactification circle.
Setting $\tilde{x} = x \prod_{\alpha = 1}^n \left( a_\alpha /b_\alpha \right)^{\frac{1}{2}}$, we may rewrite
\begin{align}
    \tilde{x} \prod_{j (\neq i)}^k \frac{[\xi_i - \xi_j - \hbar]}{[\xi_i - \xi_j + \hbar]} = \prod_{\alpha = 1}^n \frac{[\xi_i - \lambda_\alpha - s_\alpha]}{[\xi_i - \lambda_\alpha + s_\alpha]} \, ,
\end{align}
which agrees with the Bethe Ansatz Equation (BAE) for the $\mathfrak{sl}_2$-XXZ spin chain of $n$ sites in the $k$ magnon sector.
The cotangent bundle condition, $b_\alpha = a_\alpha / t$, implies $s_\alpha = \hbar/2$, which corresponds to the spin-$\tfrac{1}{2}$ system.
Such a connection between gauge theories and quantum integrable systems is known as the Bethe/gauge correspondence~\cite{Nekrasov:2009uh}.
Since $\lim_{\beta \to 0} [z]/\beta = z$, it reduces to the BAE of the $\mathfrak{sl}_2$-XXX spin chain in the 2d limit, $\beta \to 0$.

\section{Symplectic Grassmannian GLSM}\label{sec:SG}

Let $X = \operatorname{SG}(k,2n)$ be the symplectic grassmannian of $k$-dimensional subspaces of $\mathbb{C}^{2n}$ which are isotropic with respect to the symplectic form on $\mathbb{C}^{2n}$.%
\footnote{$\operatorname{SG}(n,2n)$ is also called the Lagrangian grassmannian.}
The GLSM realization was proposed in \cite{Gu:2020oeb} (see also \cite{Okonek:2019lwg}), i.e., the $\mathrm{U}(k)$ gauge theory with $2n$ fundamental chiral multiplets $\Phi_{\pm \alpha} \in \mathbb{C}^k$ for $\alpha=1,\ldots,n$, and one chiral multiplet in the dual anti-symmetric representation $Q \in \wedge^2 (\mathbb{C}^k)^\vee$, together with the superpotential given by
\begin{align}
    W_X = \sum_{\alpha=1}^n \sum_{1 \le i, j \le k} Q_{ij} \Phi_{+\alpha}^i \Phi^j_{-\alpha} \, . \label{eq:superpotential_SG}
\end{align}
Let $(a_{\pm \alpha})_{\alpha = 1,\ldots,n}$ and $\nu$ be flavor fugacities for $\Phi_{\pm \alpha}$ and $Q$.
In order that the superpotential is invariant under the flavor rotation, we have $\nu a_{+\alpha} a_{-\alpha} = 1$ for each $\alpha = 1, \ldots, n$.
Moreover, from this superpotential, we have the following relation for the R-charges of $\Phi_{\pm \alpha}$ and $Q$, $2 R_\Phi + R_Q = 2$, and rescale the fugacities, $(q^{R_\Phi/2}a_{+\alpha},q^{R_\Phi/2}a_{-\alpha},q^{R_Q/2}\nu) \mapsto (a_{+\alpha},a_{-\alpha},\nu)$, so that $\nu a_{+\alpha} a_{-\alpha} = q$.
We put $\nu = q$ and $a_{+\alpha} = a_{-\alpha}^{-1} = a_\alpha$ without loss of generality. 
Then, we have the following integral formula for the partition function on $\mathbb{D}^2 \times \mathbb{S}^1$,
\begin{align}
    Z_X = \frac{1}{k!} \int \prod_{1 \le i \neq j \le k} \left(\frac{z_i}{z_j};q\right)_\infty \prod_{1 \le i < j \le k} (z_i z_j;q)_\infty \prod_{i=1}^k \frac{x^{-\log_q z_i} \theta(z_i)^\kappa}{\prod_{\beta=1}^n (z_i/a_\beta,z_i a_\beta;q)_\infty} \frac{\dd{z_i}}{2\pi\ii z_i} \, .
\end{align}
We impose the Neumann boundary condition for the fundamental chiral multiplets, while the Dirichlet boundary condition for the dual anti-symmetric chiral multiplet,%
\footnote{The dual anti-symmetric chiral multiplet contribution with the Neumann boundary condition is $\prod_{i < j} (\nu z_i^{-1} z_j^{-1};q)_\infty^{-1}$, while it is given by $\prod_{i < j} (q \nu^{-1} z_i z_j;q)_\infty$ for the Dirichlet boundary condition. Putting $\nu = q$, we obtain the expression above.}
which is chosen to provide an additional contribution required for the K-theoretic I-function of the symplectic grassmannian via the quantum Lefschetz theorem.

We set $|x| < 1$.
In this case, we take the residues of poles located at $a_\alpha q^{-m}$ and $a_\alpha^{-1} q^{-m}$ for $m \in \mathbb{Z}_{\ge 0}$, and we denote the corresponding contour by $\mathcal{C}^+_\alpha$ and $\mathcal{C}^-_\alpha$, respectively.
Hence, we have two choices of the contour for each $\alpha \in \{1,\ldots,n\}$.
Let $\mathsf{I} : (1,\ldots,k) \hookrightarrow (1,\ldots,n)$ be an injective map as before, and we also define $\varepsilon : (1,\ldots,k) \to \{\pm 1\}$.
We write $\mathsf{I}_\varepsilon = (\mathsf{I},\varepsilon)$ and $\mathcal{C}_{\mathsf{I}_\varepsilon} = \prod_{i=1}^k \mathcal{C}_{\mathsf{I}(i)}^{\varepsilon_i}$.
We remark that, for $i \neq j$, the residue of poles at $(z_i,z_j) = (a_\alpha q^{-m_i},a_\alpha^{-1} q^{-m_j})$ is zero due to the anti-symmetric tensor contribution in the numerator.
Therefore, the number of distinct contours is given by $2^k \binom{n}{k}$, which agrees with the Euler characteristics of $\operatorname{SG}(k,2n)$.

Evaluating the integral with the contour $\mathcal{C}_{\mathsf{I}_\varepsilon}$, we obtain the Higgs branch formula of the partition function for $X = \operatorname{SG}(k,2n)$,
\begin{align}
    Z_{X,\mathsf{I}_\varepsilon} & = Z^\text{per}_{X,\mathsf{I}_\varepsilon} Z^\text{vor}_{X,\mathsf{I}_\varepsilon}
\end{align}
where
\begin{subequations}
\begin{align}
    Z^\text{per}_{X,\mathsf{I}_\varepsilon} & = \prod_{1 \le i \neq j \le k} (a_i/a_j;q)_\infty \prod_{1 \le i < j \le k} (a_i a_j;q)_\infty \prod_{i=1}^k \frac{x^{-\log_q a_i} \theta(a_i)^\kappa}{(q;q)_\infty \prod_{\beta (a_\beta \neq a_i)}^n (a_i/a_\beta;q)_\infty \prod_{\beta (a_\beta^{-1} \neq a_i)}^n (a_i a_\beta;q)_\infty} \\
    Z^\text{vor}_{X,\mathsf{I}_\varepsilon} & = \sum_{d_1,\ldots,d_k \ge 0} x^{|d|} \prod_{1 \le i \neq j \le k} \frac{\left( q^{-d_i+d_j} a_i/a_j;q\right)_{\infty}}{\left(a_i/a_j;q\right)_{\infty}} \prod_{1 \le i < j \le k} (q^{-1} a_i a_j;q^{-1})_{d_i+d_j} \prod_{i=1}^k \frac{q^{-\kappa\binom{d_i+1}{2}} (-a_i)^{\kappa d_i}}{\prod_{\beta=1}^n (q^{-1} a_i / a_\beta,q^{-1} a_i a_\beta;q^{-1})_{d_i}} \, .
\end{align}
\end{subequations}
The vortex partition function $(1-q)Z^\text{vor}_{X,\mathsf{I}_\varepsilon}$ is identified with the equivariant I-function of $X = \operatorname{SG}(k,2n)$ with the level structure with respect to $V^\vee$ under the replacement $q \mapsto q^{-1}$ \eqref{eq:I-func_SG_V*}. 
We may rewrite it as follows,
\begin{multline}
    Z^\text{vor}_{X,\mathsf{I}_\varepsilon}  = \sum_{d_1,\ldots,d_k \ge 0} x^{|d|} \prod_{1 \le i \neq j \le k} \frac{\prod_{\ell = -\infty}^{d_i - d_j} (1 - q^{-\ell} a_i/a_j)}{\prod_{\ell = -\infty}^{0} (1 - q^{-\ell} a_i/a_j)} \prod_{1 \le i < j \le k} \prod_{\ell = 1}^{d_i + d_j} (1 - q^{-\ell} a_i a_j) \\ \times \prod_{i=1}^k \frac{q^{-\kappa \binom{d_i+1}{2}} (-a_i)^{\kappa d_i}}{\prod_{\beta = 1}^n \prod_{\ell=1}^{d_i} (1 - q^{-\ell} a_i/a_\beta)(1 - q^{-\ell} a_i a_\beta)} \, .
\end{multline}
Compared with the vortex partition function of $\operatorname{Gr}(k,n)$ shown in \eqref{eq:Gr_vortex-func}, there is an additional contribution of the anti-symmetric tensor field, which prevents the factorization of the partition function even with the theta term compensating the U(1) anomaly. 
In the abelian case $k = 1$, we have
\begin{align}
    Z_{X,\alpha_\pm} & = \int_{\mathcal{C}_\alpha^\pm} \frac{x^{-\log_q z} \theta(z)^\kappa}{\prod_{\beta}^n (z/a_\beta,z a_\beta;q)_\infty} \frac{\dd z}{2 \pi \ii z} \nonumber \\
    & = \frac{x^{-\log_q a_\alpha^{\pm}} \theta(a_\alpha^{\pm})^\kappa}{(q;q)_\infty \prod_{\beta (a_\beta \neq a_\alpha^\pm)}^n (a_\alpha^\pm/a_\beta;q)_\infty \prod_{\beta (a_\beta^{-1} \neq a_\alpha^\pm)}^n (a_\alpha^\pm a_\beta;q)_\infty} \sum_{d \ge 0} \frac{x^d q^{-\kappa \binom{d+1}{2}} (-a_\alpha)^{\pm\kappa d}}{\prod_{\beta = 1}^n \prod_{\ell=1}^{d} (1 - q^{-\ell} a_\alpha^\pm/a_\beta)(1 - q^{-\ell} a_\alpha^\pm a_\beta)} \, ,
\end{align}
which agrees with the partition function of $\mathbb{P}^{2n-1}$ with the flavor fugacities, $(a_1,\ldots,a_n,a_1^{-1},\ldots,a_n^{-1})$.

\subsection{$T^*\!\operatorname{SG}(k,2n)$}\label{sec:TSG}

The GLSM for the cotangent bundle of the symplectic grassmannian $T^*\!\operatorname{SG}(k,2n)$ is given by the $\mathcal{N} = 4$ uplift of the $\operatorname{SG}(k,2n)$ GLSM.
We consider the Coulomb branch localization formula as follows,
\begin{align}
    Z_X = \frac{1}{k!} \int \prod_{1 \le i \neq j \le k} \frac{(z_i/z_j;q)_\infty}{(tz_i/z_j;q)_\infty} \prod_{1 \le i < j \le k} \frac{(z_i z_j;q)_\infty}{(u z_i z_j;q)_\infty} \prod_{\substack{i=1,\ldots,k \\ \alpha = 1, \ldots, n}} \frac{(vz_i/a_\alpha, vz_i a_\alpha;q)_\infty}{(z_i/a_\alpha,z_i a_\alpha;q)_\infty} \prod_{i=1}^k x^{-\log_q z_i} \frac{\dd{z_i}}{2\pi\ii z_i} \, .
\end{align}
In order to form $\mathcal{N} = 4$ hypermultiplets, we impose the conditions, $t = u = v$.
Under this condition, we may apply the same classification of the integration contours as before:
The additional factors in the denominator, $(t z_i/z_j;q)_\infty$ and $(u z_i z_j;q)_\infty$, do not yield any further poles contributing to the partition function when $t = u = v$.
The contour integral associated with the map $\mathsf{I}_\varepsilon = (\mathsf{I},\varepsilon)$ is then given by
\begin{align}
    Z_{X,\mathsf{I}_\varepsilon} = Z^\text{per}_{X,\mathsf{I}_\varepsilon} Z^\text{vor}_{X,\mathsf{I}_\varepsilon} \, ,
\end{align}
where
\begin{subequations}
\begin{align}
    Z^\text{per}_{X,\mathsf{I}_\varepsilon} & = \prod_{1 \le i \neq j \le k} \frac{(a_i/a_j;q)_\infty}{(ta_i/a_j;q)_\infty} \prod_{1 \le i < j \le k} \frac{(a_i a_j;q)_\infty}{(u a_i a_j;q)_\infty} \prod_{i=1}^k \frac{x^{-\log_q a_i} \prod_{\beta=1}^n (v a_i/a_\beta,v a_i a_\beta;q)_\infty}{(q;q)_\infty \prod_{\beta (a_\beta \neq a_i)}^n (a_i/a_\beta;q)_\infty \prod_{\beta (a_\beta^{-1} \neq a_i)}^n (a_i a_\beta;q)_\infty} \\    
    Z^\text{vor}_{X,\mathsf{I}_\varepsilon} & = \sum_{d_1,\ldots,d_k \ge 0} x^{|d|} \prod_{1 \le i \neq j \le k} \frac{\left( q^{-d_i+d_j} a_i/a_j;q\right)_{\infty}}{\left(t q^{-d_i+d_j} a_i/a_j;q\right)_{\infty}} \frac{\left( t a_i/a_j;q\right)_{\infty}}{\left(a_i/a_j;q\right)_{\infty}} \prod_{1 \le i < j \le k} \frac{(q^{-1} a_i a_j;q^{-1})_{d_i+d_j}}{(uq^{-1} a_i a_j;q^{-1})_{d_i+d_j}} \nonumber \\
    & \qquad \times \prod_{i=1}^k \prod_{\beta=1}^n \frac{(v q^{-1} a_i/a_\beta,vq^{-1} a_i a_\beta;q^{-1})_{d_i}}{(q^{-1} a_i / a_\beta,q^{-1} a_i a_\beta;q^{-1})_{d_i}} \, .
\end{align}    
\end{subequations}
Then, we identify $(1-q)Z^\text{vor}_{X,\mathsf{I}_\varepsilon}$ with the I-function of the odd tangent bundle $X = \Pi T\!\operatorname{SG}(k,2n)$~\eqref{eq:I-func_TSG} when $t = u = v$,
\begin{multline}
    Z^\text{vor}_{X,\mathsf{I}_\varepsilon}  = \sum_{d_1,\ldots,d_k \ge 0} x^{|d|} \prod_{1 \le i \neq j \le k} \frac{\prod_{\ell = -\infty}^{d_i - d_j} (1 - q^{-\ell} a_i/a_j)}{\prod_{\ell = -\infty}^{0} (1 - q^{-\ell} a_i/a_j)} \frac{\prod_{\ell = -\infty}^{0} (1 - t q^{-\ell} a_i/a_j)}{\prod_{\ell = -\infty}^{d_i - d_j} (1 - t q^{-\ell} a_i/a_j)} \\ \times \prod_{1 \le i < j \le k} \prod_{\ell = 1}^{d_i + d_j} \frac{1 - q^{-\ell} a_i a_j}{1 - t q^{-\ell} a_i a_j}\prod_{i=1}^k \prod_{\beta = 1}^n \prod_{\ell=1}^{d_i} \frac{(1 - t q^{-\ell} a_i/a_\beta)(1 - t q^{-\ell} a_i a_\beta)}{ (1 - q^{-\ell} a_i/a_\beta)(1 - q^{-\ell} a_i a_\beta)} \, .
\end{multline}
As discussed in \autoref{sec:TGr}, we need to assign different boundary conditions to realize the I-function of the cotangent bundle $X = T^*\!\operatorname{SG}(k,2n)$.
We consider the following integral,
\begin{multline}
    Z_X = \frac{1}{k!} \int \prod_{1 \le i \neq j \le k} \left(\frac{z_i}{z_j},\frac{q}{t}\frac{z_i}{z_j};q\right)_\infty \prod_{1 \le i < j \le k} \left(z_i z_j, \frac{q}{u} z_i^{-1} z_j^{-1};q\right)_\infty \\ \times \prod_{i=1}^k \frac{x^{-\log_q z_i}}{\prod_{\alpha=1}^n \left(z_i/a_\alpha, z_i a_\alpha,\frac{q}{v} a_\alpha/z_i,\frac{q}{v}z_i^{-1}a_\alpha^{-1};q\right)_\infty} \frac{\dd{z_i}}{2\pi\ii z_i} \, .
\end{multline}
For $|x| < 1$, we cannot take poles at $v^{-1} a_\alpha^{\pm} q^{\mathbb{Z}_{>0}}$.
Hence, evaluating the residue as before, we obtain
\begin{subequations}
\begin{align}
    Z^\text{per}_{X,\mathsf{I}_\varepsilon} & = \prod_{1 \le i \neq j \le k} (a_i/a_j,\frac{q}{t}a_i/a_j;q)_\infty \prod_{1 \le i < j \le k} (a_i a_j,\frac{q}{u}a_i^{-1}a_j^{-1};q)_\infty \nonumber \\
    & \qquad \times \prod_{i=1}^k \frac{x^{-\log_q a_i}}{(q;q)_\infty \prod_{\beta (a_\beta \neq a_i)}^n (a_i/a_\beta;q)_\infty \prod_{\beta (a_\beta^{-1} \neq a_i)}^n (a_i a_\beta;q)_\infty \prod_{\beta=1}^n (\frac{q}{v} a_\beta/a_i,\frac{q}{v} a_i^{-1} a_\beta^{-1};q)_\infty} \\    
    Z^\text{vor}_{X,\mathsf{I}_\varepsilon} & = \sum_{d_1,\ldots,d_k \ge 0} x^{|d|} \prod_{1 \le i \neq j \le k} \frac{\left( q^{-d_i+d_j} a_i/a_j;q\right)_{\infty}}{\left(a_i/a_j;q\right)_{\infty}} \frac{\left(t^{-1} q^{-d_i+d_j+1} a_i/a_j;q\right)_{\infty}}{\left( t^{-1} q a_i/a_j;q\right)_{\infty}} \nonumber \\
    & \qquad \times \prod_{1 \le i < j \le k} \frac{(q^{-1} a_i a_j; q^{-1})_{d_i+d_j}}{(u^{-1} q a_i^{-1} a_j^{-1}; q)_{d_i+d_j}} \prod_{i=1}^k \prod_{\beta=1}^n \frac{(v^{-1} q a_\beta/a_i,v^{-1} q a_i^{-1} a_\beta^{-1};q)_{d_i}}{(q^{-1} a_i / a_\beta,q^{-1} a_i a_\beta;q^{-1})_{d_i}} \nonumber \\
    & = \sum_{d_1,\ldots,d_k \ge 0} x^{|d|} \prod_{1 \le i \neq j \le k} \frac{\prod_{\ell = -\infty}^{d_i-d_j} (1 - q^{-\ell} a_i/a_j) (1 - t^{-1} q^{-\ell+1} a_i/a_j)}{\prod_{\ell = -\infty}^0 (1 - q^{-\ell} a_i/a_j) (1 - t^{-1} q^{-\ell+1} a_i/a_j)} \nonumber \\
    & \qquad \times \prod_{1 \le i < j \le k} \prod_{\ell = 1}^{d_i+d_j} \frac{1-q^{-\ell} a_i a_j}{1 - u^{-1} q^{\ell} a_i^{-1} a_j^{-1}} \prod_{i=1}^k \prod_{\beta=1}^n \prod_{\ell = 1}^{d_i} \frac{(1-v^{-1}q^{\ell}a_\beta/a_i)(1-v^{-1} q^{\ell} a_i^{-1} a_\beta^{-1})}{(1-q^{-\ell}a_i/a_\beta)(1-q^{-\ell}a_i a_\beta)} \, .
\end{align}    
\end{subequations}
Then, the vortex partition function $(1-q) Z^\text{vor}_{X,\mathsf{I}_\varepsilon}$ is identified with the localized I-function for $X = T^*\!\operatorname{SG}(k,2n)$~\eqref{eq:I-func_T*SG} under the replacement $q \mapsto q^{-1}$ and the identification of the cotangent weight $q/t = q/u = q/v$.

\subsection{Semiclassical analysis}\label{sec:SG_semiclassical}

\paragraph{$\operatorname{SG}(k,2n)$}
In the limit $q \to 1$, the partition function for $X = \operatorname{SG}(k,2n)$ behaves as follows,
\begin{align}
    Z_X = \frac{1}{k!} \int \exp \left( \frac{1}{\epsilon} \left(\widetilde{W}_X(z) + O(\epsilon) \right) \right) \prod_{i=1}^k \frac{\dd{z_i}}{2\pi\ii z_i} \, ,
\end{align}
The effective twisted superpotential $\widetilde{W}_X$ is given by 
\begin{align}
    \widetilde{W}_X(z) = \log x \sum_{i=1}^k \log z_i + \frac{\kappa}{2} \sum_{i=1}^k \log^2 (-z_i) - \sum_{1 \le i \neq j \le k} \operatorname{Li}_2(z_i/z_j) - \sum_{1 \le i < j \le k} \operatorname{Li}_2(z_i z_j) + \sum_{i=1}^k \sum_{\alpha = 1}^n \left( \operatorname{Li}_2(z_i/a_\alpha) + \operatorname{Li}_2(z_i a_\alpha) \right) \, ,
\end{align}
where we omit the constant term.
The twisted F-term equations are thus given by
\begin{align}
    -x \frac{(-z_i)^{\kappa_\text{eff}}}{z_1 \cdots z_k} \prod_{j (\neq i)}^k (1 - z_i z_j) = \prod_{\alpha = 1}^n (1 - z_i / a_\alpha)(1 - z_i a_\alpha) \, , \qquad i = 1, \ldots, k \, ,
\end{align}
where $\kappa_\text{eff} = \kappa + k$ as before.
Imposing the theta term $\theta(\det z)$ in the integrand, the corresponding twisted superpotential is $\widetilde{W}_{X}^\theta(z) = \widetilde{W}_{X}(z) + \frac{1}{2} \log^2 (-z_1 \cdots z_k) + \frac{\pi^2}{6} $, and hence we have
\begin{align}
    x (-z_i)^{\kappa_\text{eff}} \prod_{j (\neq i)}^k (1 - z_i z_j) = \prod_{\alpha = 1}^n (1 - z_i / a_\alpha)(1 - z_i a_\alpha) \, , \qquad i = 1, \ldots, k \, ,
\end{align}
from which we obtain the CS level window,
\begin{align}
    0 \le \kappa_\text{eff} \le 2n - k + 1 \iff - k \le \kappa \le 2n-2k+1 \, .
\end{align}
In the absence of the theta term $\theta(\det z)$, this window is shifted by one.
These equations generate the ideal defining the equivariant quantum K-ring of $\operatorname{SG}(k,2n)$.
We remark that these relations are slightly different from those obtained by Gu, Mihalcea, Sharpe, and Zou~\cite{Gu:2020zpg}, where they assign the Neumann boundary condition for the anti-symmetric tensor field instead of the Dirichlet boundary condition.

%\paragraph{2d limit}
Writing $z_i = \exp(\beta \xi_i)$, $a_\alpha = \exp (\beta m_\alpha)$, and taking the 2d limit, $\beta \to 0$, keeping $\tilde{x} = - x \beta^{k-2n-1}$ finite, the twisted F-term equation is reduced to
\begin{align}
    \tilde{x} \prod_{j (\neq i)}^k (\xi_i + \xi_j) = \prod_{\alpha = 1}^n (\xi_i^2 - m_\alpha^2) \, ,
\end{align}
which agrees with the equivariant chiral ring relation of $\operatorname{SG}(k,2n)$ given in \cite[eq.~(2.43)]{Gu:2020oeb}.

\paragraph{$\Pi T\!\operatorname{SG}(k,2n)$}

We turn to the semiclassical analysis of the odd tangent bundle, $X = \Pi T\!\operatorname{SG}(k,2n)$.
In this case, the twisted superpotential reads,
\begin{multline}
    \widetilde{W}_X(z) = \log x \sum_{i=1}^k \log z_i - \sum_{1 \le i \neq j \le k} \left( \operatorname{Li}_2(z_i/z_j) - \operatorname{Li}_2(t z_i/z_j) \right) - \sum_{1 \le i < j \le k} \left( \operatorname{Li}_2(z_i z_j) - \operatorname{Li}_2(u  z_i z_j) \right) \\
    + \sum_{\alpha=1}^n \sum_{i=1}^k \left( \operatorname{Li}_2(z_i/a_\alpha) + \operatorname{Li}_2(z_i a_\alpha) - \operatorname{Li}_2(v z_i/a_\alpha) - \operatorname{Li}_2(v z_i a_\alpha) \right) \, ,
\end{multline}
and the corresponding twisted F-term equations are given by
\begin{align}
    x \prod_{j (\neq i)}^k \frac{z_i - t z_j}{t z_i - z_j} \frac{1 - z_i z_j}{1 - u z_i z_j} = \prod_{\alpha = 1}^n \frac{(1 - z_i / a_\alpha) (1 - z_i a_\alpha)}{(1 - v z_i / a_\alpha) (1 - v z_i a_\alpha)} \, , \qquad i = 1, \ldots, k \, .
\end{align}
Set $t = u = v = \exp(\beta \hbar)$, $z_i = \exp(\beta(\xi_i - \hbar/2))$, $a_\alpha = \exp(\beta\lambda_\alpha)$, and $\tilde{x} = x t^{n - \frac{k-1}{2}}$.
Then, obtain
\begin{align}
    \tilde{x} \prod_{j (\neq i)}^k \frac{[\xi_i - \xi_j - \hbar]}{[\xi_i - \xi_j + \hbar]}\frac{[\xi_i + \xi_j - \hbar]}{[\xi_i + \xi_j]} = \prod_{\alpha = 1}^n \frac{[\xi_i - \lambda_\alpha - \hbar/2][\xi_i + \lambda_\alpha - \hbar/2]}{[\xi_i - \lambda_\alpha + \hbar/2][\xi_i + \lambda_\alpha + \hbar/2]} \, .
\end{align}
This equation looks similar, but differs from the twisted F-term equation for 3d SO/Sp gauge theory, which is identified with the BAE for the open spin chain with the specific boundary condition~\cite{Kimura:2020bed}.
The BAE for the spin-$\frac{1}{2}$ open chain with the diagonal boundary condition parametrized by $\{\eta_+,\eta_-\}$ is given by
\begin{align}
    B(\xi_i;\eta_+,\eta_-) \prod_{j (\neq i)}^k \frac{[\xi_i - \xi_j - \hbar]}{[\xi_i - \xi_j + \hbar]}\frac{[\xi_i + \xi_j - \hbar]}{[\xi_i + \xi_j + \hbar]} = \prod_{\alpha = 1}^n \frac{[\xi_i - \lambda_\alpha - \hbar/2][\xi_i + \lambda_\alpha - \hbar/2]}{[\xi_i - \lambda_\alpha + \hbar/2][\xi_i + \lambda_\alpha + \hbar/2]} \, , \label{eq:open_BAE}
\end{align}
where
\begin{align}
    B(\xi;\eta_+,\eta_-) = \frac{[\xi + \eta_+ - \hbar/2]}{[\xi - \eta_+ + \hbar/2]} \frac{[\xi + \eta_- - \hbar/2]}{[\xi - \eta_- + \hbar/2]} \label{eq:open_BAE_B} \, .
\end{align}
Compared to this expression, (i) the contribution of the anti-symmetric tensor field does not reproduce that appearing in the BAE, and (ii) the boundary contribution is simply given by the constant $\tilde{x}$.

For (i), one may add another hypermultiplet giving the following factor in the Coulomb branch formula, 
\begin{align}
    \prod_{1 \le i < j \le k} \frac{(z_i z_j;q)_\infty}{(t z_i z_j;q)_\infty} \frac{(t z_i z_j;q)_\infty}{(t^2 z_i z_j;q)_\infty} = \prod_{1 \le i < j \le k} \frac{(z_i z_j;q)_\infty}{(t^2 z_i z_j;q)_\infty} \, .
\end{align}
Then, the resulting twisted F-term equation reproduces the corresponding part of the BAE, $\prod_{j(\neq i)}^k [\xi_i + \xi_j - \hbar]/[\xi_i + \xi_j + \hbar]$ with $\tilde{x} = x t^{n - k + 1}$.
In this case, we may have additional residue contributions to the contour integral due to this hypermultiplet factor. 

For (ii), since we obtain this BAE from the $\mathrm{U}(k)$ gauge theory, we have the FI parameter, which plays a role of the twist parameter for the periodic spin chain.
For the SO/Sp gauge theory, however, there is no FI term, so that we obtain the open spin chain as discussed in~\cite{Kimura:2020bed}.
We may put $\tilde{x} = 1$, i.e., $x = t^{k-1-n}$.
If $|t^{k-1-n}| < 1$, the Higgs branch formula of the partition function still converges. 

\section{Orthogonal Grassmannian GLSM}\label{sec:OG}

Let $X = \operatorname{OG}(k,2n+\chi)$ be the orthogonal grassmannian with $\chi \in \{0,1\}$.
In this case, the GLSM is given by the $\mathrm{U}(k)$ gauge theory with $2n+\chi$ fundamental chiral multiplets and one chiral multiplet in the dual symmetric representation, $\operatorname{Sym}^2 (\mathbb{C}^k)^\vee$.
The contour integral formula for the partition function on $\mathbb{D}^2 \times \mathbb{S}^1$ is given by
\begin{align}
    Z_X = \frac{1}{k!} \int \prod_{1 \le i \neq j \le k} \left(\frac{z_i}{z_j};q\right)_\infty \prod_{1 \le i \le j \le k} (z_i z_j;q)_\infty \prod_{i=1}^k \frac{x^{-\log_q z_i} \theta(z_i)^\kappa}{(z_i;q)_\infty^\chi \prod_{\beta=1}^n (z_i/a_\beta,z_i a_\beta;q)_\infty} \frac{\dd{z_i}}{2\pi\ii z_i} \, . \label{eq:OG_int}
\end{align}
As in the case of $\operatorname{SG}(k,2n)$, we assign flavor fugacities in pairs, and we impose the Dirichlet boundary condition for the symmetric tensor field, while the Neumann boundary condition is assigned for the fundamental chiral fields.

We set $|x| < 1$.
For $\chi = 1$, additional poles from  $(z_i;q)_\infty$ in the denominator will be canceled by the diagonal contribution of the symmetric tensor, $(z_i^2;q)_\infty$.
We can apply the same residue classification as in the case of $\operatorname{SG}(k,2n)$ parametrized by $\mathsf{I}_\varepsilon = (\mathsf{I},\varepsilon)$ for the both cases $\chi = 0, 1$, hence we have $2^k \binom{n}{k}$ distinct contours, which agrees with the Euler characteristics of $\operatorname{OG}(k,2n+\chi)$.
Evaluating the integral with the contour $\mathcal{C}_{\mathsf{I}_\varepsilon}$, we obtain the Higgs branch formula of the partition function for $X = \operatorname{OG}(k,2n+\chi)$,
\begin{align}
    Z_{X,\mathsf{I}_\varepsilon} & = Z^\text{per}_{X,\mathsf{I}_\varepsilon} Z^\text{vor}_{X,\mathsf{I}_\varepsilon}
\end{align}
where
\begin{subequations}
\begin{align}
    Z^\text{per}_{X,\mathsf{I}_\varepsilon} & = \prod_{1 \le i \neq j \le k} (a_i/a_j;q)_\infty \prod_{1 \le i \le j \le k} (a_i a_j;q)_\infty \prod_{i=1}^k \frac{x^{-\log_q a_i} \theta(a_i)^\kappa}{(q;q)_\infty (a_i;q)_\infty^\chi \prod_{\beta (a_\beta \neq a_i)}^n (a_i/a_\beta;q)_\infty \prod_{\beta (a_\beta^{-1} \neq a_i)}^n (a_i a_\beta;q)_\infty} \\
    Z^\text{vor}_{X,\mathsf{I}_\varepsilon} & = \sum_{d_1,\ldots,d_k \ge 0} x^{|d|} \prod_{1 \le i \neq j \le k} \frac{\left( q^{-d_i+d_j} a_i/a_j;q\right)_{\infty}}{\left(a_i/a_j;q\right)_{\infty}} \prod_{1 \le i \le j \le k} (q^{-1}a_ia_j;q^{-1})_{d_i+d_j} \nonumber \\
    & \qquad \times \prod_{i=1}^k \frac{q^{-\kappa\binom{d_i+1}{2}} (-a_i)^{\kappa d_i}}{(q^{-1}a_i;q^{-1})_{d_i}^\chi \prod_{\beta=1}^n (q^{-1} a_i / a_\beta,q^{-1} a_i a_\beta;q^{-1})_{d_i}} \, .
\end{align}
\end{subequations}
The vortex partition function $(1-q)Z^\text{vor}_{X,\mathsf{I}_\varepsilon}$ is identified with the equivariant I-function of $X = \operatorname{OG}(k,2n+\chi)$ with the level structure with respect to $V^\vee$ \eqref{eq:I-func_OG_V*} under the replacement $q \mapsto q^{-1}$.
We may rewrite it as follows,
\begin{multline}
    Z^\text{vor}_{X,\mathsf{I}_\varepsilon}  = \sum_{d_1,\ldots,d_k \ge 0} x^{|d|} \prod_{1 \le i \neq j \le k} \frac{\prod_{\ell = -\infty}^{d_i - d_j} (1 - q^{-\ell} a_i/a_j)}{\prod_{\ell = -\infty}^{0} (1 - q^{-\ell} a_i/a_j)} \prod_{1 \le i \le j \le k} \prod_{\ell = 1}^{d_i + d_j} (1 - q^{-\ell} a_i a_j) \\ \times \prod_{i=1}^k \frac{q^{-\kappa \binom{d_i+1}{2}} (-a_i)^{\kappa d_i}}{\prod_{\ell=1}^{d_i} \left( (1 - q^{-\ell} a_i)^\chi \prod_{\beta = 1}^n  (1 - q^{-\ell} a_i/a_\beta)(1 - q^{-\ell} a_i a_\beta) \right)} \, .
\end{multline}
As in the case of $\operatorname{SG}(k,2n)$, we do not have the factorization of the partition function due to the additional contribution of the symmetric tensor field. 

\subsection{$\operatorname{OG}(1,2n+\chi) = Q^{2n+\chi - 2}$}\label{sec:Q}

The case $k = 1$ of the orthogonal grassmannian $X=\operatorname{OG}(1,2n+\chi)$ is identified with the quadric of dimension $2n+\chi - 2$ denoted by $Q^{2n+\chi - 2}$.
The integral formula is given by
\begin{align}
    Z_X = \int \frac{(z^2;q)_\infty x^{-\log_q z} \theta(z)^\kappa}{(z;q)_\infty^\chi \prod_{\beta=1}^n (z/a_\beta,z a_\beta;q)_\infty} \frac{\dd{z}}{2 \pi \ii z} \, .
\end{align}
In this case, there are $2n$ distinct contours, $\{\mathcal{C}_{\alpha}^\varepsilon\}_{\alpha = 1, \ldots, n,\varepsilon = \pm}$, providing $2n$ independent solutions to the following $q$-difference equation,
\begin{align}
    \left[ q^{\kappa D} (1-q^{-D})^\chi \prod_{\alpha=1}^n (1 - q^{-D}/a_\alpha)(1 - q^{-D} a_\alpha) - (-1)^\kappa x (1 - q^{-2D-1}) (1 - q^{-2D-2}) \right] y(x) = 0 \, .
\end{align}
Evaluating the integral, we obtain the Higgs branch formula of the partition function,
\begin{align}
    Z_{X,\alpha_{\pm}} & = \int_{\mathcal{C}_\alpha^\pm} \frac{(z^2;q)_\infty x^{-\log_q z} \theta(z)^\kappa}{(z;q)_\infty^\chi \prod_{\beta=1}^n (z/a_\beta,z a_\beta;q)_\infty}  \frac{\dd{z}}{2 \pi \ii z} \nonumber \\
    & = \frac{(a_\alpha^{\pm 2};q)_\infty x^{-\log_q a_\alpha^{\pm}} \theta(a_\alpha^{\pm})^\kappa}{(q;q)_\infty (a_\alpha^\pm;q)_\infty^\chi \prod_{\beta (a_\beta \neq a_\alpha^\pm)}^n (a_\alpha^\pm/a_\beta;q)_\infty \prod_{\beta (a_\beta^{-1} \neq a_\alpha^\pm)}^n (a_\alpha^\pm a_\beta;q)_\infty} \nonumber \\
    & \qquad \times \sum_{d \ge 0} \frac{x^d q^{-\kappa \binom{d+1}{2}} (-a_\alpha)^{\pm \kappa d} (q^{-1} a_\alpha^{\pm 2};q^{-1})_{2d}}{(q^{-1} a_\alpha^\pm;q^{-1})_d^\chi \prod_{\beta = 1}^n (q^{-1} a_\alpha^{\pm}/a_\beta, q^{-1} a_\alpha^\pm a_\beta;q^{-1})_d} \, .
\end{align}
Although there is a $q$-shifted factorial of degree $2d$ in the numerator, it can be written as a $q$-hypergeometric function using the identity, $(c;q)_{2d} = (\sqrt{c},-\sqrt{c},\sqrt{cq},-\sqrt{cq};q)_d$.
The vortex partition function is then given by
\begin{align}
    Z^\text{vor}_{X,\alpha^\pm} = \sum_{d \ge 0} \frac{x^d q^{-\kappa \binom{d+1}{2}} (-a_\alpha)^{\pm \kappa d} \prod_{\ell = 1}^{2d} (1-q^{-\ell} a_\alpha^{\pm 2})}{\prod_{\ell=1}^d (1-q^{-\ell} a_\alpha^\pm)^\chi \prod_{\beta = 1}^n (1-q^{-\ell} a_\alpha^{\pm}/a_\beta) (1 - q^{-\ell} a_\alpha^\pm a_\beta)} \, ,
\end{align}
where $(1-q) Z^\text{vor}_{X,\alpha^\pm}$ is identified with the localized I-function for $\operatorname{OG}(1,2n+\chi)$ with the level structure with respect to $V^\vee$ under the replacement $q \mapsto q^{-1}$.

\subsubsection{$\operatorname{OG}(1,3) \cong \mathbb{P}^1$}\label{sec:Q1}

We consider $\operatorname{OG}(1,3) = Q^1$, which is isomorphic to $\mathbb{P}^1$.
Let
\begin{align}
    Q^1=\left\{[X:Y:Z]\in\mathbb{P}^2 \mid XY-Z^2=0\right\}
\end{align}
be the standard realization of $Q^1$ as a smooth quadric in $\mathbb{P}^2$. The isomorphism $Q^1\cong\mathbb{P}^1$ is given by the degree-two Veronese embedding
\begin{align}
\nu:\mathbb{P}^1\longrightarrow\mathbb{P}^2, \qquad [s:t]\longmapsto [s^2:t^2:st] \, .
\end{align}
Its image is precisely $Q^1$, since $(s^2)(t^2)-(st)^2=0$. 
Therefore, $\nu$ identifies $\mathbb{P}^1$ with $Q^1$ and satisfies $\nu^{*}\mathcal{O}_{Q^1}(1) = \mathcal{O}_{\mathbb{P}^1}(2)$. 

Next, consider the torus action on $\mathbb{P}^1$, $[s:t]\longmapsto[bs:b^{-1}t]$. 
Under the Veronese embedding, this action is mapped to $[s^2:t^2:st] \longmapsto [b^2s^2:b^{-2}t^2:st]$. 
Therefore, the induced torus action on $Q^1$ has weights $(2,-2,0)$, and it is convenient to parametrize the corresponding equivariant weight by $a = b^2$.
Then, we consider the following localized I-functions for $\kappa = 0$,
\begin{subequations}
\begin{align}
    I_{\mathbb{P}^1, \mathsf{I}_\varepsilon} & = (1-q) \sum_{d \ge 0} \frac{x^d}{(q,qb^2;q)_d} = (1-q) \sum_{d \ge 0} \frac{x^d}{(q,qa;q)_d} \, , \\
    I_{Q^1, \mathsf{I}_\varepsilon} & = (1-q) \sum_{d \ge 0} \frac{x^{d/2} (q a^2;q)_{2d}}{(q,qa,qa^2;q)_d} = (1-q) \sum_{d \ge 0} \frac{x^{d/2}(aq^{1/2},-aq^{1/2},aq,-aq;q)_d}{(q,qa,qa^2;q)_d} % = (1-q) \sum_{d \ge 0} \frac{x^{d/2} (q^{d+1} a^2;q)_{d}}{(q,qa;q)_d} \, ,
\end{align}
\end{subequations}
where the $x$-variable for $Q^1$ is set to $x^{1/2}$ based on the curve class correspondence shown in~\eqref{equation: change of Novikov variables}.
Clearly the I-function for $Q^1$ has a problem since the number of the $q$-shifted factorials of the numerator is larger than that of the denominator.
We note that the localized I-function $I_{Q^1, \mathsf{I}_\varepsilon}$ above is never the localization of the small J-function $J_{Q^1, \mathsf{I}_\varepsilon}$. The latter is characterized by satisfying
\begin{itemize}
    \item[(1)] $\displaystyle \lim_{q\rightarrow \infty}\left(J_{Q^1, \mathsf{I}_\varepsilon}(q) - (1-q)\right) = 0$, and
    \item[(2)] $J_{Q^1, \mathsf{I}_\varepsilon}(q)$ is regular at $q = 0$.
\end{itemize}
For large enough $d$, the $x^d$-term in $I_{Q^1, \mathsf{I}_\varepsilon}$ is a rational function in $q$ whose numerator has larger degree than the denominator, and thus $I_{Q^1, \mathsf{I}_\varepsilon}$ does not satisfy (1).

More generally, the localized I-functions with level structures (see \eqref{eq:I-func_OG_level} for details) are given by
\begin{equation}
%\begin{split}
I_{Q^1, \mathsf{I}_\varepsilon}^{(V,\kappa)} = (1-q) \sum_{d \ge 0} \frac{x^{d/2} (q a^2;q)_{2d} \left(a^dq^{\binom{d}{2}}\right)^\kappa}{(q,qa,qa^2;q)_d} \ , \qquad
I_{Q^1, \mathsf{I}_\varepsilon}^{(V^\vee,\kappa)} = (1-q) \sum_{d \ge 0} \frac{x^{d/2} (q a^2;q)_{2d}\left(a^dq^{\binom{d+1}{2}}\right)^\kappa}{(q,qa,qa^2;q)_d} \ .
%\end{split}
\end{equation}
Neither of the two functions can be the corresponding small J-functions for any $\kappa$. Indeed, the CS level window is empty for $Q^1 = \operatorname{OG}(1,3)$ as we will see in \eqref{eq:CS_window_OG}.
For $\kappa \geq 0$, the numerators will have larger degrees than the denominators do so $I_{Q^1, \mathsf{I}_\varepsilon}^{(V,\kappa)}$ and $I_{Q^1, \mathsf{I}_\varepsilon}^{(V^\vee,\kappa)}$ do not satisfy (1); for $\kappa < 0$, there will be poles at $q=0$ so they do not satisfy (2).

\subsubsection{$\operatorname{OG}(1,4) \cong \mathbb{P}^1 \times \mathbb{P}^1$}

We then consider another example, $\operatorname{OG}(1,4) = Q^2 \cong \mathbb{P}^1 \times \mathbb{P}^1$.
Choose homogeneous coordinates of $Q^2$ as follows,
\begin{equation}
Q^2=
\left\{
[X_{00}:X_{01}:X_{10}:X_{11}]
\in\mathbb{P}^3
\;\middle|\;
X_{00}X_{11}-X_{01}X_{10}=0
\right\}.
\end{equation}
Then, the Segre embedding of $\mathbb{P}^1\times\mathbb{P}^1$ into $\mathbb{P}^3$ is given by
\begin{equation}
\begin{aligned}
\sigma:\ \mathbb{P}^1\times\mathbb{P}^1
\longrightarrow \mathbb{P}^3 \, , \quad
([s_0:s_1],[t_0:t_1]) \longmapsto [s_0t_0:s_0t_1:s_1t_0:s_1t_1]
\end{aligned}
\end{equation}
where the image is clearly $Q^2$. 
Now we have the pullback of the ambient hyperplane bundle, $\sigma^*\mathcal{O}_{Q^2}(1) = \mathcal{O}_{\mathbb{P}^1\times\mathbb{P}^1}(1,1)$.

Let the torus act on the first and second projective lines by
\begin{equation}
[s_0:s_1]
\longmapsto
[u\,s_0:u^{-1}s_1] \, , \qquad 
[t_0:t_1]
\longmapsto
[v\,t_0:v^{-1}t_1] \, .
\end{equation}
Therefore, their equivariant weights are obtained simply by multiplying
the weights of the two factors:
\begin{equation}
\begin{array}{ccccc} \toprule
\text{coordinates}&X_{00}&X_{01}&X_{10}&X_{11} \\ \midrule
\text{weights}&uv&u/v&v/u&(uv)^{-1} \\\bottomrule
\end{array}.
\end{equation}
On the orthogonal grassmannian side, we assign the weights, $(a_1,a_2,a_1^{-1},a_2^{-1})$.
After choosing the ordering in which $X_{00}$ ($X_{01}$) carries the weight
$a_1$ ($a_2$, resp.), the Segre map gives the identification,
\begin{equation}
(a_1,a_2,a_1^{-1},a_2^{-1})
=
(uv,u/v,(uv)^{-1},v/u) \, .
\end{equation}

We can check this correspondence at the level of the I-function. 
We consider the localized I-functions with $\kappa = 0$,
\begin{subequations}
\begin{align}
    I_{\mathbb{P}^1 \times \mathbb{P}^1, \mathsf{I}_\varepsilon} & = (1-q)\sum_{d,d' \ge 0} \frac{x^{d + d'}}{(q,q a_1 a_2;q)_d (q,q a_1/a_2;q)_{d'}} \, , \\
    I_{Q^2, \mathsf{I}_\varepsilon} & = (1-q)\sum_{d\ge0} \frac{x^d (q a_1^2;q)_{2d}}{(q,qa_1^2,q a_1a_2,q a_1/a_2;q)_d} = (1-q)\sum_{d\ge0} \frac{x^d (q^{d+1}a_1^2;q)_d}{(q,q a_1a_2,q a_1/a_2;q)_d} \, .
\end{align}
\end{subequations}
The I-function for $\mathbb{P}^1 \times \mathbb{P}^1$ is a normalized product of two I-functions of the equivariant parameters $a_1 a_2$ and $a_1/a_2$ with the same Novikov variable, i.e., the FI parameter $x$.
Then, by Lemma~\ref{lem:product_formula}, we obtain the following relation,
\begin{align}
    I_{\mathbb{P}^1 \times \mathbb{P}^1, \mathsf{I}_\varepsilon} = \frac{1}{(x;q)_\infty} I_{Q^2, \mathsf{I}_\varepsilon} \, .
\end{align}
This discrepancy is interpreted as a consequence of the
the correspondence between the curve classes of $Q^2$ and $\mathbb{P}^3$~\eqref{equation: change of Novikov variables}.
This discrepancy can be rewritten as
\begin{equation}
\frac{1}{(x;q)_\infty} = \exp \left( \sum_{k \ge 1} \frac{x^k}{k(1-q^k)} \right) \, .
\end{equation}
Such factor is expected in the following sense. By \cite{Givental:perm2}, the localized I-functions $I_{\mathbb{P}^1 \times \mathbb{P}^1, \mathsf{I}_\varepsilon}$ and $I_{Q^2, \mathsf{I}_\varepsilon}$ lie both in the range of the big J-function of point target space, which by \cite{Givental:perm8} is preserved by multiplication of the so-called $q$-string flow $\exp \left( \sum_{k \ge 1} \frac{\varepsilon^k}{k(1-q^k)} \right)$ for constant $\varepsilon$. Here we simply take $\varepsilon = x$.

\paragraph{Semiclassical analysis}

We compare the semiclassical behavior of $Q^2$ and $\mathbb{P}^1 \times \mathbb{P}^1$.
For $X=\mathbb{P}^1 \times \mathbb{P}^1$, we consider a sum of two effective twisted superpotentials of $\mathbb{P}^1$,
\begin{align}
    \widetilde{W}_{\mathbb{P}^1 \times \mathbb{P}^1} = \log x \log z_1 + \operatorname{Li}_2(z_1/u) + \operatorname{Li}_2(z_1 u) + \log x \log z_2 + \operatorname{Li}_2(z_2/v) + \operatorname{Li}_2(z_2 v) \, . 
\end{align}
The corresponding twisted F-term equations are given by
\begin{align}
    x = (1 - z_1 / u)(1 - z_1 u) \, , \qquad 
    x = (1 - z_2 / v)(1 - z_2 v) \, ,
\end{align}
from which we obtain a single equation for $z = z_1 z_2$,
\begin{multline}
    z^4 - (u+u^{-1})(v+v^{-1}) z^3 + (1-x) \left[(u+u^{-1})^2+(v+v^{-1})^2-2(1-x)\right] z^2 \\ - (1-x)^2(u+u^{-1})(v+v^{-1})z + (1-x)^4 = 0 \, . \label{eq:F-term_eq_P1P1}
\end{multline}

For $X=Q^2$, the twisted F-term equation is given by (see Section~\ref{sec:OG_semiclassical} for details)
\begin{align}
    x (1 - w^2)^2 = (1-w/a_1)(1-wa_1)(1-w/a_2)(1-wa_2) \, .
\end{align}
Putting $a_1 = uv$, $a_2 = u/v$, we may rewrite it as follows,
\begin{multline}
    (1-x) w^4 - (u+u^{-1})(v+v^{-1}) w^3 + \left[ (u+u^{-1})^2 + (v+v^{-1})^2 - 2(1-x) \right] w^2 - (u+u^{-1}) (v+v^{-1}) w + 1-x = 0 \, ,
\end{multline}
which is equivalent to \eqref{eq:F-term_eq_P1P1} under the identification $z = (1-x)w$.

\subsection{$T^*\!\operatorname{OG}(k,2n+\chi)$}\label{sec:TOG}

We consider the $\mathcal{N} = 4$ uplift of the orthogonal grassmannian GLSM, which describes the cotangent bundle $T^*\!\operatorname{OG}(k,2n+\chi)$.
We consider the following Coulomb branch integral formula,
\begin{multline}
    Z_X = \frac{1}{k!} \int \prod_{1 \le i \neq j \le k} \frac{(z_i/z_j;q)_\infty}{(tz_i/z_j;q)_\infty} \prod_{1 \le i \le j \le k} \frac{(z_i z_j;q)_\infty}{(u z_i z_j;q)_\infty} \\
    \times \prod_{\substack{i=1,\ldots,k \\ \alpha = 1, \ldots, n}} \frac{(v z_i/a_\alpha,v z_i a_\alpha;q)_\infty}{(z_i/a_\alpha,z_i a_\alpha;q)_\infty} \prod_{i=1}^k x^{-\log_q z_i} \frac{(vz_i;q)_\infty^\chi}{(z_i;q)_\infty^\chi} \frac{\dd{z_i}}{2\pi\ii z_i} \, .
\end{multline}
Under the condition $t = u = v$, there is no additional residue contribution compared to the previous case, so that we may apply the same classification of the integration contours as before.
We have
\begin{align}
    Z_{X,\mathsf{I}_\varepsilon} = Z^\text{per}_{X,\mathsf{I}_\varepsilon} Z^\text{vor}_{X,\mathsf{I}_\varepsilon} \, ,
\end{align}
where
\begin{subequations}
\begin{align}
    Z^\text{per}_{X,\mathsf{I}_\varepsilon} & = \prod_{1 \le i \neq j \le k} \frac{(a_i/a_j;q)_\infty}{(ta_i/a_j;q)_\infty} \prod_{1 \le i \le j \le k} \frac{(a_i a_j;q)_\infty}{(u a_i a_j;q)_\infty} \prod_{i=1}^k \frac{x^{-\log_q a_i} (v a_i;q)_\infty^\chi \prod_{\beta=1}^n (v a_i/a_\beta,v a_i a_\beta;q)_\infty}{(q;q)_\infty (a_i;q)_\infty^\chi \prod_{\beta (a_\beta \neq a_i)}^n (a_i/a_\beta;q)_\infty \prod_{\beta (a_\beta^{-1} \neq a_i)}^n (a_i a_\beta;q)_\infty} \\    
    Z^\text{vor}_{X,\mathsf{I}_\varepsilon} & = \sum_{d_1,\ldots,d_k \ge 0} x^{|d|} \prod_{1 \le i \neq j \le k} \frac{\left( q^{-d_i+d_j} a_i/a_j;q\right)_{\infty}}{\left(t q^{-d_i+d_j} a_i/a_j;q\right)_{\infty}} \frac{\left( t a_i/a_j;q\right)_{\infty}}{\left(a_i/a_j;q\right)_{\infty}} \prod_{1 \le i \le j \le k} \frac{(q^{-1}a_ia_j;q^{-1})_{d_i+d_j}}{(uq^{-1}a_ia_j;q^{-1})_{d_i+d_j}} \nonumber \\
    & \qquad \times \prod_{i=1}^k \frac{(vq^{-1}a_i;q^{-1})_{d_i}^\chi}{(q^{-1}a_i;q^{-1})_{d_i}^\chi} \prod_{\beta=1}^n \frac{(v q^{-1} a_i/a_\beta,vq^{-1} a_i a_\beta;q^{-1})_{d_i}}{(q^{-1} a_i / a_\beta,q^{-1} a_i a_\beta;q^{-1})_{d_i}} \, .
\end{align}    
\end{subequations}
We identify $(1-q)Z^\text{vor}_{X,\mathsf{I}_\varepsilon}$ with the localized I-function for $X = \Pi T\!\operatorname{OG}(k,2n+\chi)$~\eqref{eq:I-func_TSG} under the map $q \mapsto q^{-1}$ and setting $t = u = v$, 
\begin{multline}
    Z^\text{vor}_{X,\mathsf{I}_\varepsilon}  = \sum_{d_1,\ldots,d_k \ge 0} x^{|d|} \prod_{1 \le i \neq j \le k} \frac{\prod_{\ell = -\infty}^{d_i - d_j} (1 - q^{-\ell} a_i/a_j)}{\prod_{\ell = -\infty}^{0} (1 - q^{-\ell} a_i/a_j)} \frac{\prod_{\ell = -\infty}^{0} (1 - t q^{-\ell} a_i/a_j)}{\prod_{\ell = -\infty}^{d_i - d_j} (1 - t q^{-\ell} a_i/a_j)} \\ \times \prod_{1 \le i \le j \le k} \prod_{\ell = 1}^{d_i + d_j} \frac{1 - q^{-\ell} a_i a_j}{1 - u q^{-\ell} a_i a_j}\prod_{i=1}^k \prod_{\ell=1}^{d_i} \left( \frac{1-v q^{-\ell} a_i}{1-q^{-\ell} a_i} \right)^\chi \prod_{\beta = 1}^n \frac{(1 - v q^{-\ell} a_i/a_\beta)(1 - v q^{-\ell} a_i a_\beta)}{ (1 - q^{-\ell} a_i/a_\beta)(1 - q^{-\ell} a_i a_\beta)} \, .
\end{multline}
In order to obtain the I-function of the cotangent bundle $X = T^*\!\operatorname{OG}(k,2n+\chi)$, as discussed in \autoref{sec:TGr}, we assign different boundary conditions, which yield the following integral, 
\begin{multline}
    Z_X = \frac{1}{k!} \int \prod_{1 \le i \neq j \le k} \left(\frac{z_i}{z_j},\frac{q}{t} \frac{z_i}{z_j};q\right)_\infty \prod_{1 \le i \le j \le k} \left(z_i z_j,\frac{q}{u}z_i^{-1}z_j^{-1};q\right)_\infty \\
    \times \prod_{i=1}^k \frac{x^{-\log_q z_i}}{(z_i,\frac{q}{v}z_i^{-1};q)_\infty^\chi \prod_{\alpha=1}^n (z_i/a_\alpha,z_i a_\alpha,\frac{q}{v} a_\alpha/z_i,\frac{q}{v} z_i^{-1} a_\alpha^{-1};q)_\infty} \frac{\dd{z_i}}{2\pi\ii z_i} \, .
\end{multline}
Applying the same residue analysis, we obtain
\begin{subequations}
\begin{align}
    Z^\text{per}_{X,\mathsf{I}_\varepsilon} & = \prod_{1 \le i \neq j \le k} \left(a_i/a_j,\frac{q}{t} a_i/a_j;q\right)_\infty \prod_{1 \le i \le j \le k} \left(a_i a_j,\frac{q}{u}a_i^{-1}a_j^{-1};q\right)_\infty \nonumber \\
    & \qquad \times \prod_{i=1}^k \frac{x^{-\log_q a_i}}{(q;q)_\infty (a_i,\frac{q}{v}a_i^{-1};q)_\infty^\chi \prod_{\beta (a_\beta \neq a_i)}^n (a_i/a_\beta;q)_\infty \prod_{\beta (a_\beta^{-1} \neq a_i)}^n (a_i a_\beta;q)_\infty \prod_{\beta=1}^n (\frac{q}{v} a_\beta/a_i,\frac{q}{v} a_i^{-1} a_\beta^{-1};q)_\infty} \\    
    Z^\text{vor}_{X,\mathsf{I}_\varepsilon} & = \sum_{d_1,\ldots,d_k \ge 0} x^{|d|} \prod_{1 \le i \neq j \le k} \frac{\left( q^{-d_i+d_j} a_i/a_j;q\right)_{\infty}}{\left(a_i/a_j;q\right)_{\infty}} \frac{\left(t^{-1} q^{-d_i+d_j+1} a_i/a_j;q\right)_{\infty}}{\left( t^{-1} q a_i/a_j;q\right)_{\infty}} \nonumber \\
    & \qquad \times \prod_{1 \le i \le j \le k} \frac{(q^{-1}a_ia_j;q^{-1})_{d_i+d_j}}{(u^{-1}qa_i^{-1}a_j^{-1};q)_{d_i+d_j}} \prod_{i=1}^k \frac{(v^{-1} q a_i^{-1};q)_{d_i}^\chi}{(q^{-1} a_i;q^{-1})_{d_i}^\chi} \prod_{\beta=1}^n \frac{(v^{-1} q a_\beta/a_i,v^{-1} q a_i^{-1} a_\beta^{-1};q)_{d_i}}{(q^{-1} a_i / a_\beta,q^{-1} a_i a_\beta;q^{-1})_{d_i}} \nonumber \\
    & = \sum_{d_1,\ldots,d_k \ge 0} x^{|d|} \prod_{1 \le i \neq j \le k} \frac{\prod_{\ell = -\infty}^{d_i-d_j} (1 - q^{-\ell} a_i/a_j) (1 - t^{-1} q^{-\ell+1} a_i/a_j)}{\prod_{\ell = -\infty}^0 (1 - q^{-\ell} a_i/a_j) (1 - t^{-1} q^{-\ell+1} a_i/a_j)} \nonumber \\
    & \qquad \times \prod_{1 \le i \le j \le k} \prod_{\ell = 1}^{d_i+d_j} \frac{1-q^{-\ell} a_i a_j}{1 - u^{-1} q^{\ell} a_i^{-1} a_j^{-1}} \prod_{i=1}^k \prod_{\ell = 1}^{d_i} \left( \frac{1-v^{-1} q^{\ell}a_i^{-1}}{1-q^{-\ell}a_i} \right)^\chi \prod_{\beta=1}^n  \frac{(1-v^{-1}q^{\ell}a_\beta/a_i)(1-v^{-1} q^{\ell} a_i^{-1} a_\beta^{-1})}{(1-q^{-\ell}a_i/a_\beta)(1-q^{-\ell}a_i a_\beta)} \, .
\end{align}    
\end{subequations}
We identify $(1-q) Z^\text{vor}_{X,\mathsf{I}_\varepsilon}$ with the localized I-function for $X = T^*\!\operatorname{OG}(k,2n+\chi)$~\eqref{eq:I-func_T*OG} under the replacement $q \mapsto q^{-1}$ and the identification of the cotangent weight $q/t = q/u = q/v$. 

\subsection{Semiclassical analysis}\label{sec:OG_semiclassical}

\paragraph{$\operatorname{OG}(k,2n+\chi)$}

From the integral formula~\eqref{eq:OG_int}, we obtain the effective twisted superpotential for $X = \operatorname{OG}(k,2n+\chi)$ as follows,
\begin{multline}
    \widetilde{W}_X(z) = \log x \sum_{i=1}^k \log z_i + \frac{\kappa}{2} \sum_{i=1}^k \log^2 (-z_i) - \sum_{1 \le i \neq j \le k} \operatorname{Li}_2(z_i/z_j) - \sum_{1 \le i \le j \le k} \operatorname{Li}_2(z_i z_j) \\ + \sum_{i=1}^k \sum_{\alpha = 1}^n \left( \operatorname{Li}_2(z_i/a_\alpha) + \operatorname{Li}_2(z_i a_\alpha) \right) + \chi \sum_{i=1}^n \operatorname{Li}_2(z_i) \, ,
\end{multline}
where we omit the constant term, and the twisted F-term equations are thus given by
\begin{align}
    -x \frac{(-z_i)^{\kappa_\text{eff}}}{z_1 \cdots z_k} (1 - z_i^2)^2 \prod_{j (\neq i)}^k (1 - z_i z_j) = (1 - z_i)^\chi \prod_{\alpha = 1}^n (1 - z_i / a_\alpha)(1 - z_i a_\alpha) \, , \qquad i = 1, \ldots, k \, .
\end{align}
Imposing the theta term $\theta(\det z;q)$ in the integrand, the corresponding twisted superpotential is given by $\widetilde{W}_{X}^\theta(z) = \widetilde{W}_{X}(z) + \frac{1}{2} \log^2 (-z_1 \cdots z_k) + \frac{\pi^2}{6}$, and hence we have
\begin{align}
    x (-z_i)^{\kappa_\text{eff}} (1 - z_i^2)^2 \prod_{j (\neq i)}^k (1 - z_i z_j) = (1 - z_i)^\chi \prod_{\alpha = 1}^n (1 - z_i / a_\alpha)(1 - z_i a_\alpha) \, , \qquad i = 1, \ldots, k \, ,
\end{align}
from which we obtain the CS level window,
\begin{align}
    0 \le \kappa_\text{eff} \le 2n + \chi - k - 3 \iff - k \le \kappa \le 2n - 2k + \chi - 3 \, . \label{eq:CS_window_OG}
\end{align}
In the absence of the theta term $\theta(\det z)$, this window is shifted by one.
These equations generate the ideal defining the equivariant quantum K-ring of $\operatorname{OG}(k,2n+\chi)$.
We remark that, from this analysis, we see that there is no CS window for $\operatorname{OG}(1,3)$ and $\operatorname{OG}(2,4)$.
See \autoref{sec:Q1} for details of the case $\operatorname{OG}(1,3) = Q^1$.
Writing $z_i = \exp(\beta \xi_i)$, $a_\alpha = \exp (\beta m_\alpha)$, and taking the 2d limit, $\beta \to 0$, keeping $\tilde{x} = (-1)^{\chi+1} 4 x \beta^{k-2n-\chi+1}$ finite, the twisted F-term equation is reduced to
\begin{align}
    \tilde{x} \prod_{j (\neq i)}^k (\xi_i + \xi_j) = \xi_i^{\chi-2} \prod_{\alpha = 1}^n (\xi_i^2 - m_\alpha^2) \, ,
\end{align}
which agrees with the equivariant chiral ring relation of $\operatorname{OG}(k,2n+\chi)$ given in \cite[\S3.2]{Gu:2020oeb}.

\paragraph{$\Pi T\!\operatorname{OG}(k,2n+\chi)$}

For the odd tangent bundle theory $X = \Pi T\!\operatorname{OG}(k,2n+\chi)$, the effective twisted superpotential reads,
\begin{multline}
    \widetilde{W}_X(z) = \log x \sum_{i=1}^k \log z_i - \sum_{1 \le i \neq j \le k} \left( \operatorname{Li}_2(z_i/z_j) - \operatorname{Li}_2(t z_i/z_j) \right) - \sum_{1 \le i \le j \le k} \left( \operatorname{Li}_2(z_i z_j) - \operatorname{Li}_2(u  z_i z_j) \right) \\
    + \sum_{i=1}^k \sum_{\alpha=1}^n  \left( \operatorname{Li}_2(z_i/a_\alpha) + \operatorname{Li}_2(z_i a_\alpha) - \operatorname{Li}_2(v z_i/a_\alpha) - \operatorname{Li}_2(v z_i a_\alpha) \right) + \chi \sum_{i=1}^k \left( \operatorname{Li}_2(z_i) - \operatorname{Li}_2(v z_i) \right) \, ,
\end{multline}
and the corresponding twisted F-term equations are given by
\begin{align}
    x \left( \frac{1 - z_i^2}{1 - u z_i^2} \right)^2 \prod_{j (\neq i)}^k \frac{z_i - t z_j}{t z_i - z_j} \frac{1 - z_i z_j}{1 - u z_i z_j} = \left( \frac{1 - z_i}{1 - v z_i} \right)^\chi \prod_{\alpha = 1}^n \frac{(1 - z_i / a_\alpha) (1 - z_i a_\alpha)}{(1 - v z_i / a_\alpha) (1 - v z_i a_\alpha)} \, , \qquad i = 1, \ldots, k \, .
\end{align}
Set $t = u = v = \exp(\beta \hbar)$, $z_i = \exp(\beta(\xi_i - \hbar/2))$, $a_\alpha = \exp(\beta\lambda_\alpha)$, and $\tilde{x} = x t^{n + \frac{\chi}{2} - \frac{k+1}{2}}$.
Then, we obtain
\begin{align}
    \tilde{x} \left( \frac{[2\xi_i-\hbar]}{[2\xi_i]} \right)^2 \left( \frac{[\xi_i + \hbar/2]}{[\xi_i - \hbar/2]} \right)^\chi \prod_{j (\neq i)}^k \frac{[\xi_i - \xi_j - \hbar]}{[\xi_i - \xi_j + \hbar]}\frac{[\xi_i + \xi_j - \hbar]}{[\xi_i + \xi_j]} = \prod_{\alpha = 1}^n \frac{[\xi_i - \lambda_\alpha - \hbar/2][\xi_i + \lambda_\alpha - \hbar/2]}{[\xi_i - \lambda_\alpha + \hbar/2][\xi_i + \lambda_\alpha + \hbar/2]} \, .
\end{align}
As in the case of $T^*\!\operatorname{SG}(k,2n)$, although this is close to the BAE for the open spin chain~\eqref{eq:open_BAE}, there are some mismatches. 
The contribution of the symmetric tensor field can be modified by adding another hypermultiplet in the Coulomb branch integral formula, $\prod_{1 \le i \le j \le k} (z_i z_j;q)_\infty / (t^2 z_i z_j;q)_\infty$.
The modified version of the twisted F-term equations are then given by
\begin{align}
    \tilde{x} \left( \frac{[2(\xi_i-\hbar/2)]}{[2(\xi_i + \hbar/2)]} \right)^2 \left( \frac{[\xi_i + \hbar/2]}{[\xi_i - \hbar/2]} \right)^\chi \prod_{j (\neq i)}^k \frac{[\xi_i - \xi_j - \hbar]}{[\xi_i - \xi_j + \hbar]} \frac{[\xi_i + \xi_j - \hbar]}{[\xi_i + \xi_j + \hbar]} = \prod_{\alpha = 1}^n \frac{[\xi_i - \lambda_\alpha - \hbar/2][\xi_i + \lambda_\alpha - \hbar/2]}{[\xi_i - \lambda_\alpha + \hbar/2][\xi_i + \lambda_\alpha + \hbar/2]} \, ,
\end{align}
with $\tilde{x} = x t^{n + \chi/2 - k - 1}$.
However, it is not yet identified with the BAE due to the first term as in the case of the 3d Sp gauge theory~\cite{Kimura:2020bed}. 
In the 2d limit, $\beta \to 0$, we instead obtain
\begin{align}
    \tilde{x} \left( \frac{\xi_i-\hbar/2}{\xi_i + \hbar/2} \right)^{2 - \chi} \prod_{j (\neq i)}^k \frac{(\xi_i - \xi_j - \hbar)}{(\xi_i - \xi_j + \hbar)} \frac{(\xi_i + \xi_j - \hbar)}{(\xi_i + \xi_j + \hbar)} = \prod_{\alpha = 1}^n \frac{(\xi_i - \lambda_\alpha - \hbar/2)(\xi_i + \lambda_\alpha - \hbar/2)}{(\xi_i - \lambda_\alpha + \hbar/2)(\xi_i + \lambda_\alpha + \hbar/2)} \, .
\end{align}
Compared to the boundary contribution~\eqref{eq:open_BAE_B}, we may identify
\begin{align}
    \{\eta_+,\eta_-\} = 
    \begin{cases}
        \{0,0\} & (\chi = 0) \\
        \{0,\infty\} & (\chi = 1) \\
    \end{cases}
\end{align}
together with $x = - t^{k + 1 - n - \chi}$.

\appendix
\section{K-theoretic I-functions}\label{sec:I-func}

\subsection{Generating functions in quantum K-theory}
For smooth projective variety $X$, let $K(X)$ be its K-theory with complex coefficients.
Its big J-function $\mathcal{J}^X$ is the generating function defined by
\begin{equation} \label{equation: def of big J-function}
\mathcal{J}^X(\t(q)) := 1 - q + \t(q) + \sum_{\substack{m\geq 0,i \\ \beta\in H_2(X;\mathbb{Z})}} \ x^\beta \cdot \phi^i \langle \frac{\phi_i}{1-qL_0}, \t(L_1), \cdots, \t(L_m)\rangle^{\mathfrak{S}_m}_{0,m+1,\beta} \ \in K(X)(q)\llbracket x\rrbracket 
\end{equation}
where 
\begin{itemize}
    \item the input $\t = \t(q)$ lives in $K(X)[q,q^{-1}]\llbracket x\rrbracket $,
    \item $\{\phi_i\}$ and $\{\phi^i\}$ are bases of $K(X)$ dual under the K-theoretic Poincaré pairing, i.e. $\chi(X;\phi_i\otimes \phi^j) = \delta_i^j$,
    \item $x$ refers to the Novikov variables which record the curve degrees, and
    \item the correlators $\langle \cdots \rangle_{0,m+1,\beta}^{\mathfrak{S}_m}$ are the permutation-equivariant K-theoretic Gromov-Witten invariants, or quantum K invariants, defined in \cite{Givental:perm1} as Euler characteristics of bundles over the moduli space of stable maps $\M_{0,m+1}(X,\beta)$.
\end{itemize}
The special point $J^X := \mathcal{J}^X(0) \in K(X)(q)\llbracket x\rrbracket $ is called the small J-function. 

The definition of $\langle \cdots \rangle_{0,m+1,\beta}^{\mathfrak{S}_m}$ depends on the so-called virtual structure sheaf $\mathcal{O}^{\mathrm{virt}}_{0,m+1,\beta}$ defined in \cite{Lee:2004}. When we modify these virtual structure sheaves, we obtain twisted quantum K invariants and thus twisted generating functions.

\paragraph{Level structure}
One type of modification that we consider is the level structure introduced in \cite{RuanZhangWen:2026}. More precisely, given vector bundle $E\rightarrow X$ and integer $\kappa \in\mathbb{Z}$, we consider the modification
\begin{equation} \label{equation: level structure on sheaves}
\mathcal{O}^{\mathrm{virt}}_{0,m+1,\beta} \mapsto \mathcal{O}^{\mathrm{virt}}_{0,m+1,\beta} \otimes {\det}^{-\kappa}(E_{0,m+1,\beta}),
\end{equation}
where $E_{0,m+1,\beta}:= \mathrm{ft}_{*}\mathrm{ev}^*E$ refers to the virtual bundle over $\mathcal{O}^{\mathrm{virt}}_{0,m+1,\beta}$ whose fiber over the stable map $f: (\Sigma,p_1,\cdots,p_{m+1})\rightarrow X$ is $H^0(\Sigma,f^*E)-H^1(\Sigma,f^*E)$ and $\det$ gives by definition $\det(V-W) = \det(V)\otimes \det^{-1}(W)$ on virtual bundles. The resulting invariants are called quantum K invariants with level structure $(E,\kappa)$. We denote by $\mathcal{J}^{X,(E,\kappa)} = \mathcal{J}^{X,(E,\kappa)}(\t(q))\in K(X)(q)\llbracket x\rrbracket $ the big J-function with level structure $(E,\kappa)$, defined via (\ref{equation: def of big J-function}) by replacing the quantum K invariants by the ones with level structures above, and $\{\phi^i\}$ by $\{\phi^{(E,\kappa), i}\}$ the basis satisfying $\chi(X;\phi_i\otimes \phi^{(E,\kappa), j}\otimes \det^{-\kappa}E) = \delta_i^j$. We denote by $J^{X,(E,\kappa)}:= \mathcal{J}^{X,(E,\kappa)}(0)$ the small J-function with level structure $(E,\kappa)$.

\paragraph{Vector bundle structure}
Another type of modification that we consider gives quantum K invariants of a vector bundle over $X$. More precisely, given vector bundle $E\rightarrow X$, we assume that $E$ admits a fiber-wise $\C^\times$-action with weight $v^{-1}$, for $v$ the equivariant parameter of $\C^\times$. Then, we consider the modification
\begin{equation}
\mathcal{O}^{\mathrm{virt}}_{0,m+1,\beta} \mapsto \mathcal{O}^{\mathrm{virt}}_{0,m+1,\beta} \otimes \mathrm{Eu}^{-1}(v^{-1}E_{0,m+1,\beta}),
\end{equation}
where $\mathrm{Eu}$ refers to the K-theoretic Euler class and we use $v^{-1}$ to emphasize the fiber-wise $\C^\times$-action. The resulting invariants are indeed the ($\C^\times$-equivariant) quantum K-invariants of $\mathrm{Tot}(E)$. We denote by $\mathcal{J}^E = \mathcal{J}^E(\t(q))\in K_{\C^\times}(X)(q)\llbracket x\rrbracket $ the big J-function of $\mathrm{Tot}(E)$, defined via (\ref{equation: def of big J-function}) by replacing the quantum K invariants by the ones of $\mathrm{Tot}(E)$ above, and $\{\phi^i\}$ by $\{\phi^{E, i}\}$ the basis satisfying $\chi(X;\phi_i\otimes \phi^{E, j}\otimes \mathrm{Eu}^{-1}(v^{-1}E)) = \delta_i^j$. Here $K_{\C^\times}(X) = K(X)(v)$ refers to the $\C^\times$-equivariant K-theory of $X$, with coefficients localized from $\mathrm{Rep}(\C^\times) = \C[v]$ to $\mathrm{Frac}(\mathrm{Rep}(\C^\times)) = \C(v)$. We denote by $J^{E}:= \mathcal{J}^{E}(0)$ the small J-function of $\mathrm{Tot}(E)$. 

\paragraph{Odd vector bundle structure}
Another type of modification very similar to above will, on the other hand, help with understanding quantum K-invariants of subvarieties of $X$ under certain non-equivariant limits, as we will see below. To distinguish from the case above, given vector bundle $E\rightarrow X$, we assume it is equipped with a fiber-wise $\mathbb{C}^\times$-action with weight $u^{-1}$, for $u$ the equivariant parameter of $\mathbb{C}^\times$. Then, we consider the modification on the virtual structure sheaves
\begin{equation} \label{equation: Euler twisting on sheaves}
\mathcal{O}^{\mathrm{virt}}_{0,m+1,\beta} \mapsto \mathcal{O}^{\mathrm{virt}}_{0,m+1,\beta} \otimes \mathrm{Eu}(u^{-1}E_{0,m+1,\beta}).
\end{equation}
We denote by $\mathcal{J}^{E[1]} = \mathcal{J}^{E[1]}(\t(q))\in K_{\C^\times}(X)(q)\llbracket x\rrbracket $ the $(\mathrm{Eu},E)$-twisted big J-function, or the big J-function of the odd vector bundle $E[1]$, defined via (\ref{equation: def of big J-function}) by replacing the quantum K invariants by the third type of modified invariants above, and $\{\phi^i\}$ by $\{\phi^{E[1], i}\}$ the basis satisfying $\chi(X;\phi_i\otimes \phi^{E[1], j}\otimes \mathrm{Eu}(u^{-1}E)) = \delta_i^j$. Here $K_{\C^\times}(X) = K(X)(u)$ refers again to the $\C^\times$-equivariant K-theory of $X$ with localized coefficients. We denote by $J^{E[1]}:= \mathcal{J}^{E[1]}(0)$ the small J-function of $E[1]$.

When an action by torus $T$ exists, the entire story above admits an equivariant version.

Below, we give explicit formulas for special values of various twisted torus-equivariant big J-functions of $X$ for $X$ being grassmannians, symplectic grassmannians, and orthogonal grassmannians. We call these special values the I-functions. They mostly take $q$-hypergeometric form and are thus more directly related to integral models. They differ from the corresponding small J-functions, which are usually hard to compute, by changes of coordinates known as the ``mirror maps'' \cite{ZhangZhou:2020}.

\subsection{$\mathrm{Gr}(k,n)$}\label{sec:I-func_Gr}
We consider first the case of $X = \mathrm{Gr}(k,n)$ the grassmannian of dim-$k$ subspaces of $\C^n$. $H_2(X;\mathbb{Z}) = \mathbb{Z}$, and we denote the only Novikov variable by $x$. The standard action of $T = (\C^\times)^n$ on $\C^n$ induces a $T$-action on $X$. We denote by $a_1,\cdots,a_n$ the equivariant parameters of $T$. By \cite{Givental:2021dbg}, the I-function
\begin{equation}\label{eq:I-func_Gr}
I^X := (1-q) \sum_{d_1,\ldots,d_k \ge 0} \frac{x^{d_1+\cdots+d_k}}{\prod_{i=1}^k \prod_{\alpha=1}^n \prod_{\ell=1}^{d_i}(1-q^\ell P_i/a_\alpha)} \prod_{1 \le i, j \le k} \frac{\prod_{\ell=-\infty}^{d_i-d_j}(1-q^\ell P_i/P_j)}{\prod_{\ell=-\infty}^{0}(1-q^\ell P_i/P_j)}, 
\end{equation}
where $P_1,\cdots,P_k$ denote the K-theoretic Chern roots of the tautological (sub-)bundle $V\rightarrow X$ (such that $\Lambda^i V$ is represented by the $i$-th elementary symmetric polynomial in $P_1,\cdots,P_k$), is a special value of $\mathcal{J}^X$. Indeed, $I^X = J^X$ in this case and the mirror map is trivial.

The above formula is obtained via the method of abelian/non-abelian correspondence. More precisely, we write $X$ as a GIT quotient $X = \mathrm{Hom}(\C^k,\C^n)/\!\!/\mathrm{GL}(k)= \mathrm{Gr}(k,n)$, and consider an associated toric variety (abelian quotient) $Y = \mathrm{Hom}(\C^k,\C^n)/\!\!/(\mathbb{C}^\times)^k = (\mathbb{P}^{n-1})^k$. $\mathcal{J}^X$ is proven to be obtained from $\mathcal{J}^Y$ via a certain transformation, and $I^X$ is exactly the image of the well-known I-function $I^Y$ of $Y$ under this transformation. The same method works for the twisted big J-functions as well. In particular, the I-function with level structure $(V,\kappa)$
\begin{equation}\label{eq:I-func_Gr_V}
I^{X,(V,\kappa)} := (1-q) \sum_{d_1,\ldots,d_k \ge 0} \frac{x^{d_1+\cdots+d_k} \prod_{i=1}^k\left(P_i^{d_i}q^{\frac{d_i(d_i-1)}{2}}\right)^\kappa}{\prod_{i=1}^k \prod_{\alpha=1}^n \prod_{\ell=1}^{d_i}(1-q^\ell P_i/a_\alpha)} \prod_{1 \le i, j \le k} \frac{\prod_{\ell=-\infty}^{d_i-d_j}(1-q^\ell P_i/P_j)}{\prod_{\ell=-\infty}^{0}(1-q^\ell P_i/P_j)}
\end{equation}
is a special value of $\mathcal{J}^{X,(V,\kappa)}$, the I-function with level structure $(V^\vee,\kappa)$
\begin{equation}\label{eq:I-func_Gr_V*}
I^{X,(V^\vee,\kappa)} := (1-q) \sum_{d_1,\ldots,d_k \ge 0} \frac{x^{d_1+\cdots+d_k} \prod_{i=1}^k\left(P_i^{d_i}q^{\frac{d_i(d_i+1)}{2}}\right)^\kappa}{\prod_{i=1}^k \prod_{\alpha=1}^n \prod_{\ell=1}^{d_i}(1-q^\ell P_i/a_\alpha)} \prod_{1 \le i, j \le k} \frac{\prod_{\ell=-\infty}^{d_i-d_j}(1-q^\ell P_i/P_j)}{\prod_{\ell=-\infty}^{0}(1-q^\ell P_i/P_j)},
\end{equation}
is a special value of $\mathcal{J}^{X,(V^\vee,\kappa)}$, the I-function of $T^*\!X$%
\footnote{%
By Lemma~\ref{lem:q-shifted_ratio}, we may rewrite the second line as follows,
\begin{align*}
    \prod_{1 \le i, j \le k} \frac{\prod_{\ell=-\infty}^{d_i-d_j}(1-q^\ell P_i/P_j)}{\prod_{\ell=-\infty}^{0}(1-q^\ell P_i/P_j)} \frac{\prod_{\ell=-\infty}^{-1}(1-vq^{-\ell} P_j/P_i)}{\prod_{\ell=-\infty}^{d_i-d_j-1}(1-vq^{-\ell} P_j/P_i)} 
    & = \prod_{1 \le i, j \le k} \frac{\prod_{\ell=-\infty}^{d_i-d_j}(1-q^\ell P_i/P_j)}{\prod_{\ell=-\infty}^{0}(1-q^\ell P_i/P_j)} \frac{\prod_{\ell=-\infty}^{d_i-d_j}(1-vq^{\ell} P_i/P_j)}{\prod_{\ell=-\infty}^{0}(1-vq^{\ell} P_i/P_j)} \, .
\end{align*}
}
\begin{equation}\label{eq:I-func_T*Gr}
\begin{split}
I^{T^*\!X} := & (1-q) \sum_{d_1,\ldots,d_k \ge 0} x^{d_1+\cdots+d_k} \cdot \prod_{i=1}^k \prod_{\alpha=1}^n\frac{\prod_{\ell=0}^{d_i-1}(1-v q^{-\ell} a_\alpha/P_i)}{\prod_{\ell=1}^{d_i}(1-q^\ell P_i/a_\alpha)} \\
& \cdot \prod_{1 \le i, j \le k} \frac{\prod_{\ell=-\infty}^{d_i-d_j}(1-q^\ell P_i/P_j)}{\prod_{\ell=-\infty}^{0}(1-q^\ell P_i/P_j)} \frac{\prod_{\ell=-\infty}^{-1}(1-vq^{-\ell} P_j/P_i)}{\prod_{\ell=-\infty}^{d_i-d_j-1}(1-vq^{-\ell} P_j/P_i)}
\end{split}
\end{equation}
is a special value of $\mathcal{J}^{T^*X}$, and the I-function of $TX[1]$
\begin{equation}\label{eq:I-func_TGr}
\begin{split}
I^{TX[1]} := & (1-q) \sum_{d_1,\ldots,d_k \ge 0} x^{d_1+\cdots+d_k} \cdot \prod_{i=1}^k \prod_{\alpha=1}^n\frac{\prod_{\ell=1}^{d_i}(1-u q^\ell P_i/a_\alpha)}{\prod_{\ell=1}^{d_i}(1-q^\ell P_i/a_\alpha)} \\
& \cdot \prod_{1 \le i, j \le k} \frac{\prod_{\ell=-\infty}^{d_i-d_j}(1-q^\ell P_i/P_j)}{\prod_{\ell=-\infty}^{0}(1-q^\ell P_i/P_j)} \frac{\prod_{\ell=-\infty}^{0}(1-u     q^\ell P_i/P_j)}{\prod_{\ell=-\infty}^{d_i-d_j}(1-u q^\ell P_i/P_j)}
\end{split}
\end{equation}
is a special value of $\mathcal{J}^{TX[1]}$.

\subsection{Quantum Lefschetz}
The symplectic grassmannians and orthogonal grassmannians are zero loci of vector bundles over the ordinary grassmannians. We aim to compute their I-functions using the quantum Lefschetz theorem for grassmannians developed in \cite{Givental:2021dbg}.

We start with the general set-up where $E\rightarrow X$ is a convex ($H^1(\mathbb{P}^1,f^*E)=0, \ \forall f: \mathbb{P}^1\rightarrow X$) vector bundle over a smooth projective variety, and consider $\mathcal{J}^{E[1]}$ the $(\mathrm{Eu},E)$-twisted big J-function. Since $E$ is convex, there exists short exact sequence of vector bundles over $\M_{0,m+1}(X;\beta)$
\begin{equation}
0\rightarrow E_{0,m+1,\beta}' \rightarrow E_{0,m+1,\beta} \rightarrow \mathrm{ev}_0^* E\rightarrow 0
\end{equation}
where $\mathrm{ev}_0: \M_{0,m+1}(X;\beta)\rightarrow X$ is the evaluation at the first marked point and thus $\mathrm{ev}_0^* E$ has fiber $E|_{p_0}$ over the stable map $f:(\Sigma,p_0,\cdots,p_{m}) \rightarrow X$. Then, it is not hard to see that the non-equivariant limit at $u=1$ of $\mathcal{J}^{E[1]}$ exists, and that the limit $\lim_{u\rightarrow 1}\mathcal{J}^{E[1]}$ is given still by (\ref{equation: def of big J-function}), but with the quantum K invariants replaced by the ones defined by the modified virtual structure sheaves
\begin{equation}
\mathcal{O}^{\mathrm{virt}}_{0,m+1,\beta} \mapsto \mathcal{O}^{\mathrm{virt}}_{0,m+1,\beta} \otimes \mathrm{Eu}(E_{0,m+1,\beta}'),
\end{equation}
where $\mathrm{Eu}$ is now the non-equivariant K-theoretic Euler class, and with $\{\phi^i\}$ kept as they are.

Now let $\iota: Z\rightarrow X$ be the zero locus of a regular section $s: X\rightarrow E$. It is not hard to see that 
\begin{equation}
\bigsqcup_{\substack{\gamma\in H_2(Z;\mathbb{Z}) \\ \text{s.t. }i_*\gamma = \beta}} \M_{0,m+1}(Z;\gamma) \quad \subset \quad \M_{0,m+1}(X;\beta)
\end{equation}
is the zero locus of a regular section of $E_{0,m+1,\beta}$ defined by restricting $s$ to the corresponding curve. Then, by results in \cite{YanKimura:2026qLefvp}, we have
\begin{equation} \label{equation: qLef on J function level}
\mathcal{J}^Z(\iota^*\t(q))|_{x_Z^\gamma\mapsto x_X^{\iota_*\gamma}} = \iota^* \lim_{u\rightarrow 1}\mathcal{J}^{E[1]}(\t(q)).
\end{equation}
Note that (\ref{equation: qLef on J function level}) is true under torus-equivariant settings as well, where we ask $s: X \rightarrow E$ to be a torus-invariant section.

\subsection{$\mathrm{SG}(k,2n)$}
For $k\leq n$, we consider $Z = \mathrm{SG}(k,2n)$, the grassmannian of dim-$k$ isotropic subspaces of the symplectic space $(\C^{2n},\omega = \sum_{i=1}^n dz_i\wedge dz_{n+i})$. Consider the convex vector bundle $E = \Lambda^2 V^\vee\rightarrow X = \mathrm{Gr}(k,2n)$ where $V$ is the rank-$k$ tautological (sub-)bundle. Then, $Z$ is exactly the zero locus of the regular section $s: X\rightarrow E$ given by the restriction of $\omega$ to the corresponding subspace. We denote still by $V$ the restriction of the tautological bundle from $X$ to $Z$. 

$\iota_*: H_2(Z;\mathbb{Z}) = \mathbb{Z} \rightarrow H_2(X;\mathbb{Z}) = \mathbb{Z}$ is the identity map in this case, so we may continue to use $x$ to denote the Novikov variable of $Z$.

The $T = (\C^\times)^n$ action on $\C^{2n}$ by 
\begin{equation}
(t_1,\cdots,t_n) \cdot (z_1,\cdots,z_{2n}) = (t_1z_1,\cdots,t_nz_{n}, t_1^{-1}z_{n+1},\cdots,t_n^{-1}z_{2n})
\end{equation}
preserves the symplectic form, and thus reduces to an action on $X$ and on $Z$. On $Z$, this is exactly the action of the maximal torus of $G = \mathrm{Sp}(2n)$, if we regard $Z = G/P$ as a homogeneous space. We denote the equivariant parameters by $a_1,\cdots,a_n$.

By applying the non-abelian quantum Lefschetz theorem developed in \cite{Yan2024} (to the special case of grassmannians), we can obtain a special value of $\mathcal{J}^{E[1]}$. Taking the non-equivariant limit at $u=1$, we obtain a special value of $\iota^* \lim_{u\rightarrow 1}\mathcal{J}^{E[1]}$ and thus of $\mathcal{J}^Z$ by (\ref{equation: qLef on J function level}), which we call the I-function of $Z$
\begin{equation}\label{eq:I-func_SG}
I^Z = (1-q) \sum_{d_1,\ldots,d_k \ge 0} \frac{x^{d_1+\cdots+d_k} \cdot \prod_{1\leq i < j\leq k} \prod_{\ell=1}^{d_i+d_j}(1-q^\ell P_iP_j)}{\prod_{i=1}^k \prod_{\alpha=1}^n \prod_{\ell=1}^{d_i}(1-q^\ell P_i/a_\alpha)(1-q^\ell P_ia_\alpha)} \prod_{1 \le i, j \le k} \frac{\prod_{\ell=-\infty}^{d_i-d_j}(1-q^\ell P_i/P_j)}{\prod_{\ell=-\infty}^{0}(1-q^\ell P_i/P_j)}.
\end{equation}

Now consider the level structure $(V,\kappa)$ on $Z$. Since $\iota_\gamma^* \det^{-\kappa}(V_{0,m+1,\beta}) = \det^{-\kappa}(V_{0,m+1,\gamma})$, similar to (\ref{equation: qLef on J function level}), one can prove
\begin{equation}
\mathcal{J}^{Z,(V,\kappa)}(\iota^*\t(q))|_{x_Z^\gamma\mapsto x_X^{\iota_*\gamma}} = \iota^* \lim_{u\rightarrow 1}\mathcal{J}^{E[1],(V,\kappa)}(\t(q)),
\end{equation}
where $\mathcal{J}^{E[1],(V,\kappa)}$ is the generating function of quantum K invariants from virtual structure sheaves twisted both by the level structure $(V,\kappa)$ defined in (\ref{equation: level structure on sheaves}) and the Euler class (\ref{equation: Euler twisting on sheaves}). By the non-abelian quantum Lefschetz theorem and the formula for level structure developed in \cite{Yan2024}, we obtain a special value of $\mathcal{J}^{E[1],(V,\kappa)}$, and thus a special value of $\mathcal{J}^{Z,(V,\kappa)}$ at $u=1$, which we call the I-function of $Z$ with level structure $(V,\kappa)$
\begin{equation}\label{eq:I-func_SG_V}
\begin{split}
I^{Z,(V,\kappa)} = & (1-q) \sum_{d_1,\ldots,d_k \ge 0} \frac{x^{d_1+\cdots+d_k} \cdot \prod_{1\leq i < j\leq k} \prod_{\ell=1}^{d_i+d_j}(1-q^\ell P_iP_j) \cdot \prod_{i=1}^k\left(P_i^{d_i}q^{\frac{d_i(d_i-1)}{2}}\right)^\kappa}{\prod_{i=1}^k \prod_{\alpha=1}^n \prod_{\ell=1}^{d_i}(1-q^\ell P_i/a_\alpha)(1-q^\ell P_ia_\alpha)} \\
& \cdot \prod_{1 \le i, j \le k} \frac{\prod_{\ell=-\infty}^{d_i-d_j}(1-q^\ell P_i/P_j)}{\prod_{\ell=-\infty}^{0}(1-q^\ell P_i/P_j)},
\end{split}
\end{equation}
Similarly, we obtain the I-function of $Z$ with level structure $(V^\vee,\kappa)$
\begin{equation}\label{eq:I-func_SG_V*}
\begin{split}
I^{Z,(V^\vee,\kappa)} = & (1-q) \sum_{d_1,\ldots,d_k \ge 0} \frac{x^{d_1+\cdots+d_k} \cdot \prod_{1\leq i < j\leq k} \prod_{\ell=1}^{d_i+d_j}(1-q^\ell P_iP_j) \cdot \prod_{i=1}^k\left(P_i^{d_i}q^{\frac{d_i(d_i+1)}{2}}\right)^\kappa}{\prod_{i=1}^k \prod_{\alpha=1}^n \prod_{\ell=1}^{d_i}(1-q^\ell P_i/a_\alpha)(1-q^\ell P_ia_\alpha)} \\
& \cdot \prod_{1 \le i, j \le k} \frac{\prod_{\ell=-\infty}^{d_i-d_j}(1-q^\ell P_i/P_j)}{\prod_{\ell=-\infty}^{0}(1-q^\ell P_i/P_j)}.
\end{split}
\end{equation}

For the theory of $T^*Z$, we note that $\iota_\gamma^* \mathrm{Eu}^{-1}(v^{-1}(T^*X-\Lambda^2V)_{0,m+1,\beta}) = \mathrm{Eu}^{-1}(v^{-1}(T^*Z)_{0,m+1,\gamma})$. Hence, similar to (\ref{equation: qLef on J function level}), we have
\begin{equation}
\mathcal{J}^{T^*Z}(\iota^*\t(q))|_{x_Z^\gamma\mapsto x_X^{\iota_*\gamma}} = \iota^* \lim_{u\rightarrow 1}\mathcal{J}^{T^*X-\Lambda^2V, E[1]}(\t(q)),
\end{equation}
where $\mathcal{J}^{T^*X-\Lambda^2V, E[1]}(\t(q))$ is the generating function of quantum K invariants twisted by
\begin{equation}
\mathcal{O}^{\mathrm{virt}}_{0,m+1,\beta} \mapsto \mathcal{O}^{\mathrm{virt}}_{0,m+1,\beta} \cdot \mathrm{Eu}^{-1}(v^{-1}(T^*X-\Lambda^2V)_{0,m+1,\beta}) \cdot \mathrm{Eu}(u^{-1}E_{0,m+1,\beta}).
\end{equation}
By the non-abelian quantum Lefschetz theorem, we can obtain a special value of $\mathcal{J}^{T^*X-\Lambda^2V, E[1]}$, and thus a special value of $\mathcal{J}^{T^*Z}$ at $u=1$, which we call the I-function of $T^*Z$
\begin{equation}\label{eq:I-func_T*SG}
\begin{split}
I^{T^*Z} = & (1-q) \sum_{d_1,\ldots,d_k \ge 0} x^{d_1+\cdots+d_k} \cdot \prod_{1\leq i < j\leq k} \frac{\prod_{\ell=1}^{d_i+d_j}(1-q^\ell P_iP_j)}{\prod_{\ell=0}^{d_i+d_j-1}(1-vq^{-\ell} P_i^{-1}P_j^{-1})} \\
& \cdot \prod_{i=1}^k \prod_{\alpha=1}^n\frac{\prod_{\ell=0}^{d_i-1}(1-vq^{-\ell} a_\alpha/P_i)(1-vq^{-\ell} a_\alpha^{-1}/P_i)}{\prod_{\ell=1}^{d_i}(1-q^\ell P_i/a_\alpha)(1-q^\ell P_ia_\alpha)} \\
& \cdot \prod_{1 \le i, j \le k} \frac{\prod_{\ell=-\infty}^{d_i-d_j}(1-q^\ell P_i/P_j)}{\prod_{\ell=-\infty}^{0}(1-q^\ell P_i/P_j)} \frac{\prod_{\ell=-\infty}^{-1}(1-vq^{-\ell} P_j/P_i)}{\prod_{\ell=-\infty}^{d_i-d_j-1}(1-vq^{-\ell} P_j/P_i)}.
\end{split}
\end{equation}
Similarly, for the theory of $TZ[1]$, since 
$\iota_\gamma^* \mathrm{Eu}(u^{-1}(TX-\Lambda^2V^\vee)_{0,m+1,\beta}) = \mathrm{Eu}(u^{-1}(TZ)_{0,m+1,\gamma})$, we have
\begin{equation}
\mathcal{J}^{TZ[1]}(\iota^*\t(q))|_{x_Z^\gamma\mapsto x_X^{\iota_*\gamma}} = \iota^* \lim_{u\rightarrow 1}\mathcal{J}^{(TX-\Lambda^2V^\vee)[1], E[1]}(\t(q)),
\end{equation}
where $\mathcal{J}^{(TX-\Lambda^2V^\vee)[1], E[1]}(\t(q))$ is the generating function of quantum K invariants twisted by
\begin{equation}
\mathcal{O}^{\mathrm{virt}}_{0,m+1,\beta} \mapsto \mathcal{O}^{\mathrm{virt}}_{0,m+1,\beta} \cdot \mathrm{Eu}(u^{-1}(TX-\Lambda^2V^\vee)_{0,m+1,\beta}) \cdot \mathrm{Eu}(u^{-1}E_{0,m+1,\beta}).
\end{equation}
By the non-abelian quantum Lefschetz theorem, we can obtain a special value of $\mathcal{J}^{(TX-\Lambda^2V^\vee)[1], E[1]}$, and thus a special value of $\mathcal{J}^{TZ[1]}$ at $u=1$, which we call the I-function of $TZ[1]$
\begin{equation}\label{eq:I-func_TSG}
\begin{split}
I^{TZ[1]} = & (1-q) \sum_{d_1,\ldots,d_k \ge 0} x^{d_1+\cdots+d_k} \cdot \prod_{1\leq i < j\leq k} \frac{\prod_{\ell=1}^{d_i+d_j}(1-q^\ell P_iP_j)}{\prod_{\ell=1}^{d_i+d_j}(1-u q^\ell P_iP_j)} \\
& \cdot \prod_{i=1}^k \prod_{\alpha=1}^n\frac{\prod_{\ell=1}^{d_i}(1-uq^\ell P_i/a_\alpha)(1-uq^\ell P_ia_\alpha)}{\prod_{\ell=1}^{d_i}(1-q^\ell P_i/a_\alpha)(1-q^\ell P_ia_\alpha)} \\
& \cdot \prod_{1 \le i, j \le k} \frac{\prod_{\ell=-\infty}^{d_i-d_j}(1-q^\ell P_i/P_j)}{\prod_{\ell=-\infty}^{0}(1-q^\ell P_i/P_j)} \frac{\prod_{\ell=-\infty}^{0}(1-uq^\ell P_i/P_j)}{\prod_{\ell=-\infty}^{d_i-d_j}(1-uq^\ell P_i/P_j)}.
\end{split}
\end{equation}

\subsection{$\mathrm{OG}(k,2n+\chi)$}
Lastly, we consider $Z = \mathrm{OG}(k,2n+\chi)$ where $\chi = 0,1$, the grassmannian of dim-$k$ (orthogonally) isotropic subspaces of $(\mathbb{C}^{2n+\chi},Q)$ where the symmetric form $Q$ is given by
\begin{equation}
Q(\mathbf{e}_\beta,\mathbf{e}_{n+\beta})=1\ (1\leq i\leq n), \quad Q(\mathbf{e}_{2n+1},\mathbf{e}_{2n+1})=1 \text{ if } \chi=1, \quad Q(\mathbf{e}_a,\mathbf{e}_b)=0 \text{ otherwise}.
\end{equation} 
$E = \mathrm{Sym}^2 V^\vee \rightarrow X = \mathrm{Gr}(k,2n+\chi)$ is a convex vector bundle, and $Z$ is the zero locus of its regular section $s: X\rightarrow E$ given by the restriction of $Q$ to the correspondence subspace. We denote still by $V$ the restriction of tautological bundle from $X$ to $Z$.

The $T = (\mathbb{C}^\times)^n$ action on $\mathbb{C}^{2n+\chi}$ by 
\begin{equation}
(t_1,\cdots,t_n) \cdot (z_1,\cdots,z_{2n+\chi}) = 
\begin{cases}
(t_1z_1,\cdots,t_nz_{n}, t_1^{-1}z_{n+1},\cdots,t_n^{-1}z_{2n}), & \chi=0 \\
(t_1z_1,\cdots,t_nz_{n}, t_1^{-1}z_{n+1},\cdots,t_n^{-1}z_{2n},z_{2n+1}), & \chi=1
\end{cases}
\end{equation}
preserves $Q$, and thus reduces to an action on $X$ and on $Z$. Again, on $Z$ this is the action of the maximal torus when we regard it as a homogeneous space. We denote the equivariant parameters still by $a_1,\cdots,a_n$. 

$\iota_*: H_2(Z;\mathbb{Z}) \rightarrow H_2(X;\mathbb{Z})$ is, however, much more complicated in the orthogonal case. More precisely,
\begin{equation} \label{equation: change of Novikov variables}
\iota_*: 
\begin{cases}
\mathbb{Z} \rightarrow \mathbb{Z}, 1\mapsto 1, & \quad \text{if } \chi = 1, k < n, \\
\mathbb{Z} \rightarrow \mathbb{Z}, 1\mapsto 2, & \quad \text{if } \chi = 1, k = n, \\
\mathbb{Z} \rightarrow \mathbb{Z}, 1\mapsto 1, & \quad \text{if } \chi = 0, k < n - 1, \\
\mathbb{Z}^2 \rightarrow \mathbb{Z}, (1,0)\mapsto 1, (0,1)\mapsto 1, & \quad \text{if } \chi = 0, k = n - 1, \\
\mathbb{Z}^2 \rightarrow \mathbb{Z}, (1,0)\mapsto 2, (0,1)\mapsto 2, & \quad \text{if } \chi = 0, k = n.
\end{cases}
\end{equation}
Following exactly the same method as in the symplectic case, we obtain the I-function of $Z$
\begin{equation}\label{eq:I-func_OG}
\begin{split}
I^Z = & (1-q) \sum_{d_1,\ldots,d_k \ge 0} \frac{ x^{d_1+\cdots+d_k} \cdot \prod_{1 \le i \le j \le k} \prod_{\ell = 1}^{d_i + d_j} (1 - q^{\ell} P_iP_j)}{\prod_{i=1}^k \prod_{\ell=1}^{d_i} \left( (1 - q^{\ell} P_i)^\chi \prod_{\beta = 1}^n  (1 - q^{\ell} P_i/a_\beta)(1 - q^{\ell} P_i a_\beta) \right)} \\
& \cdot \prod_{1 \le i \neq j \le k} \frac{\prod_{\ell = -\infty}^{d_i - d_j} (1 - q^{\ell} P_i/P_j)}{\prod_{\ell = -\infty}^{0} (1 - q^{\ell} P_i/P_j)},
\end{split}
\end{equation}
the I-functions of $Z$ with level structures
\begin{subequations}\label{eq:I-func_OG_level}
\begin{equation}\label{eq:I-func_OG_V}
\begin{split}
I^{Z,(V,\kappa)} = & (1-q) \sum_{d_1,\ldots,d_k \ge 0} \frac{ x^{d_1+\cdots+d_k} \cdot \prod_{1 \le i \le j \le k} \prod_{\ell = 1}^{d_i + d_j} (1 - q^{\ell} P_iP_j) \cdot \prod_{i=1}^k\left(P_i^{d_i}q^{\frac{d_i(d_i-1)}{2}}\right)^\kappa}{\prod_{i=1}^k \prod_{\ell=1}^{d_i} \left( (1 - q^{\ell} P_i)^\chi \prod_{\beta = 1}^n  (1 - q^{\ell} P_i/a_\beta)(1 - q^{\ell} P_i a_\beta) \right)} \\
& \cdot \prod_{1 \le i \neq j \le k} \frac{\prod_{\ell = -\infty}^{d_i - d_j} (1 - q^{\ell} P_i/P_j)}{\prod_{\ell = -\infty}^{0} (1 - q^{\ell} P_i/P_j)},
\end{split}
\end{equation}
\begin{equation}\label{eq:I-func_OG_V*}
\begin{split}
I^{Z,(V^\vee,\kappa)} = & (1-q) \sum_{d_1,\ldots,d_k \ge 0} \frac{ x^{d_1+\cdots+d_k} \cdot \prod_{1 \le i \le j \le k} \prod_{\ell = 1}^{d_i + d_j} (1 - q^{\ell} P_iP_j) \cdot \prod_{i=1}^k\left(P_i^{d_i}q^{\frac{d_i(d_i+1)}{2}}\right)^\kappa}{\prod_{i=1}^k \prod_{\ell=1}^{d_i} \left( (1 - q^{\ell} P_i)^\chi \prod_{\beta = 1}^n  (1 - q^{\ell} P_i/a_\beta)(1 - q^{\ell} P_i a_\beta) \right)} \\
& \cdot \prod_{1 \le i \neq j \le k} \frac{\prod_{\ell = -\infty}^{d_i - d_j} (1 - q^{\ell} P_i/P_j)}{\prod_{\ell = -\infty}^{0} (1 - q^{\ell} P_i/P_j)},
\end{split}
\end{equation}
\end{subequations}
the I-function of $T^*Z$
\begin{equation}\label{eq:I-func_T*OG}
\begin{split}
I^{T^*Z} = & (1-q) \sum_{d_1,\ldots,d_k \ge 0} x^{d_1+\cdots+d_k} \cdot \prod_{1\leq i \leq j\leq k} \frac{\prod_{\ell=1}^{d_i+d_j}(1-q^\ell P_iP_j)}{\prod_{\ell=0}^{d_i+d_j-1}(1-vq^{-\ell} P_i^{-1}P_j^{-1})} \\
& \cdot \prod_{i=1}^k \left( \frac{\prod_{\ell=0}^{d_i-1}(1-vq^{-\ell}/P_i)^\chi}{\prod_{\ell=1}^{d_i}(1-q^\ell P_i)^\chi} \prod_{\beta=1}^n\frac{\prod_{\ell=0}^{d_i-1}(1-vq^{-\ell} a_\beta/P_i)(1-vq^{-\ell} a_\beta^{-1}/P_i)}{\prod_{\ell=1}^{d_i}(1-q^\ell P_i/a_\beta)(1-q^\ell P_ia_\beta)} \right) \\
& \cdot \prod_{1 \le i, j \le k} \frac{\prod_{\ell=-\infty}^{d_i-d_j}(1-q^\ell P_i/P_j)}{\prod_{\ell=-\infty}^{0}(1-q^\ell P_i/P_j)} \frac{\prod_{\ell=-\infty}^{-1}(1-vq^{-\ell} P_j/P_i)}{\prod_{\ell=-\infty}^{d_i-d_j-1}(1-vq^{-\ell} P_j/P_i)},
\end{split}
\end{equation}
and the I-function of $TZ[1]$
\begin{equation}\label{eq:I-func_TOG}
\begin{split}
I^{TZ[1]} = & (1-q) \sum_{d_1,\ldots,d_k \ge 0} x^{d_1+\cdots+d_k} \cdot \prod_{1\leq i \leq j\leq k} \frac{\prod_{\ell=1}^{d_i+d_j}(1-q^\ell P_iP_j)}{\prod_{\ell=1}^{d_i+d_j}(1-u q^\ell P_iP_j)} \\
& \cdot \prod_{i=1}^k \left( \frac{\prod_{\ell=1}^{d_i}(1-u q^\ell P_i)^\chi}{\prod_{\ell=1}^{d_i}(1-q^\ell P_i)^\chi} \prod_{\beta=1}^n\frac{\prod_{\ell=1}^{d_i}(1-uq^\ell P_i/a_\beta)(1-uq^\ell P_ia_\beta)}{\prod_{\ell=1}^{d_i}(1-q^\ell P_i/a_\beta)(1-q^\ell P_ia_\beta)} \right) \\
& \cdot \prod_{1 \le i, j \le k} \frac{\prod_{\ell=-\infty}^{d_i-d_j}(1-q^\ell P_i/P_j)}{\prod_{\ell=-\infty}^{0}(1-q^\ell P_i/P_j)} \frac{\prod_{\ell=-\infty}^{0}(1-uq^\ell P_i/P_j)}{\prod_{\ell=-\infty}^{d_i-d_j}(1-u q^\ell P_i/P_j)},
\end{split}
\end{equation}
Here $x$ refers always to the Novikov variable of $X = \mathrm{Gr}(k,2n+\chi)$. The I-functions are all special values of the corresponding big J-functions, but under the change of Novikov variables from $Z$ to $X$ designated by (\ref{equation: change of Novikov variables}). Note that in some cases, (\ref{equation: change of Novikov variables}) can only give rise to even powers of $x$. It does not imply that the odd $x$-power terms in the I-function vanish, but rather suggests that the I-function should be regarded as the special value of an extended big J-function, where the input $\t$ is allowed to involve half-powers of $x$ as coefficients.

\section{Formulas}\label{sec:formulas}

\subsection*{$q$-shifted factorial}
We define the $q$-shifted factorial,
\begin{align}
    (z;q)_n = \prod_{m=0}^{n-1} (1 - z q^m) = \sum_{k = 0}^n (-z)^k q^{\binom{k}{2}} \binom{n}{k}_q \, ,
\end{align}
where the $q$-binomial coefficient is given by 
\begin{align}
    \binom{n}{k}_q = \frac{(q;q)_n}{(q;q)_k (q;q)_{n-k}} = \frac{(q^{-n};q)_k}{(q;q)_k} (-1)^k q^{nk - \binom{k}{2}} \, .
\end{align}
For $|q| < 1$, we have 
\begin{align}
    (z;q)_\infty = \prod_{m=0}^{\infty} (1 - z q^m) = \sum_{k = 0}^\infty \frac{(-z)^k q^{\binom{k}{2}}}{(q;q)_k} \, , \label{eq:q-Pochhammer}
\end{align}
where $\lim_{n \to \infty} \binom{n}{k}_q = 1/(q;q)_k$.
We write
\begin{align}
    (z_1,\ldots,z_k;q)_d = \prod_{i=1}^k (z_i;q)_d \, .
\end{align}

\begin{lemma}\label{lem:q-shifted_ratio}
For $d \in \mathbb{Z}$, $z, q \in \mathbb{C}^\times$, we define
\begin{align}
    f_d(z;q) := 
    \begin{cases}
        (q^{-1}z;q^{-1})_{d} & (d \ge 0) \\ 
        (z;q)_{|d|}^{-1} = (-z)^{-|d|} q^{-\binom{|d|}{2}} (z^{-1};q^{-1})_{|d|}^{-1} & (d < 0)
    \end{cases} \, .
\end{align}
Then, it is given by a ratio of infinite products,
\begin{align}
    f_d(z;q)
     = 
     \begin{cases}
        \displaystyle 
         \frac{\prod_{\ell=-\infty}^d (1-q^{-\ell}z)}{\prod_{\ell=-\infty}^0 (1-q^{-\ell}z)} = (q^{-d}z;q)_\infty/(z;q)_\infty & (|q|<1) \\[1em]
         \displaystyle
         \frac{\prod_{\ell = -\infty}^{-1}(1 - q^\ell z)}{\prod_{\ell = -\infty}^{-d-1}(1 - q^\ell z)} = (q^{-1}z;q^{-1})_\infty/(q^{-d-1}z;q^{-1})_\infty & (|q|>1)
     \end{cases} \, .
\end{align}
\end{lemma}

\begin{lemma}\label{lem:q-shifted_ratio2}
For $d \in \mathbb{Z}$, we have
\begin{align}
    \frac{(zq^{-d},z^{-1}q^d;q)_\infty}{(z,z^{-1};q)_\infty} = (-z)^d q^{-\binom{d}{2}} \frac{1 - z q^{-d}}{1 - z} \, .
\end{align}
\end{lemma}

\subsection*{Theta function}
For $|q| < 1$, the theta function is defined by
\begin{align}
    \theta(z;q) = (z,q/z;q)_\infty = \frac{1}{(q;q)_\infty} \sum_{n \in \mathbb{Z}} (-1)^n q^{n \choose 2} z^n \, , \label{eq:theta_def}
\end{align}
obeying the shift relation,
\begin{align}
    \theta(zq^m;q) = (-z)^{-m} q^{-{m \choose 2}} \theta(z;q) \, , \quad m \in \mathbb{Z} \, . \label{eq:theta_shift}
\end{align}
We write
\begin{align}
    \theta(z_1,\ldots,z_k;q) = \prod_{i=1}^k \theta(z_i;q) \, .
\end{align}
We have a reflection formula,
\begin{align}
    (z;q)_\infty (z^{-1};q)_\infty = (1 - z^{-1}) \theta(z;q) = (1 - z) \theta(z^{-1};q) \, . \label{eq:q-reflection}
\end{align}

\subsection*{Hypergeometric functions}
For $|z| < 1$, we define the hypergeometric functions,
\begin{subequations}
\begin{align}
    {}_r F_s \left( \begin{matrix} a_1,\ldots,a_r \\ b_1,\ldots,b_s \end{matrix} ; z \right) & = \sum_{m=0}^\infty \frac{(a_1)_m \cdots (a_r)_m}{(b_1)_m \cdots (b_s)_m} \frac{z^m}{m!} \, , \\
    {}_r \phi_s \left( \begin{matrix} a_1,\ldots,a_r \\ b_1,\ldots,b_s \end{matrix} ; q, z \right) & = \sum_{m=0}^\infty \left( (-1)^m q^{m \choose 2} \right)^{s-r+1} \frac{(a_1,\ldots,a_r;q)_m}{{(q,b_1,\ldots,b_s;q)_m}} z^m \, . \label{eq:q-hypergeom}
\end{align}
\end{subequations}
We have the $q$-Chu--Vandermonde identity,
\begin{align}
    {}_2 \phi_1 \left(
    \begin{matrix}
        q^{-n}, b \\ c
    \end{matrix} ; q, q \right) 
    = \sum_{m = 0}^n \frac{(q^{-n},b;q)_m}{(q,c;q)_m} q^m
    = \frac{(c/b;q)_n}{(c;q)_n} b^n \, , \label{eq:q-Chu-Vandermonde}
\end{align}
or equivalently
\begin{align}
    {}_2 \phi_1 \left(
    \begin{matrix}
        q^{-n}, b \\ c
    \end{matrix} ; q, \frac{c q^n}{b} \right) = \frac{(c/b;q)_n}{(c;q)_n} \, . \label{eq:q-Chu-Vandermonde2}
\end{align}

\begin{lemma}\label{lem:product_formula}
    Let $F(z;a) = \sum_{d \ge 0} z^d / (q,a;q)_d$. Then, we have
    \begin{align}
        F(z;a) F(z;b) (z;q)_\infty = \sum_{d \ge 0} \frac{(ab q^{d-1};q)_d}{(q,a,b;q)_d} z^d \, .
    \end{align}
\end{lemma}
\begin{proof}
    By \eqref{eq:q-Pochhammer}, we expand the LHS as follows,
    \begin{align}
        F(z;a) F(z;b) (z;q)_\infty = \sum_{d \ge 0} C_d z^d \, , 
    \end{align}
    where
    \begin{align}
        C_d & = \sum_{\substack{l,m,n \ge 0 \\ l + m + n = d}} \frac{(-1)^l q^{\binom{l}{2}}}{(q;q)_l (q,a;q)_m (q,b;q)_n} = \sum_{m = 0}^d \sum_{n = 0}^{d-m} \frac{(-1)^{d-m-n} q^{\binom{d-m-n}{2}}}{(q;q)_{d-m-n} (q,a;q)_m (q,b;q)_n} \, .
    \end{align}
    We first compute the summation over the index $n$. 
    Put $\ell = d - m$.
    Then, we have
    \begin{align}
        \sum_{n = 0}^{d-m} \frac{(-1)^{d-m-n} q^{\binom{d-m-n}{2}}}{(q;q)_{d-m-n} (q,b;q)_n} = \sum_{n = 0}^\ell \frac{(-1)^{\ell-n} q^{\binom{\ell-n}{2}}}{(q;q)_{\ell-n} (q,b;q)_n} \, .
    \end{align}
    Noticing
    \begin{align}
        \frac{1}{(q;q)_{\ell - n}} = \frac{(q^{-\ell};q)_n}{(q;q)_\ell} (-1)^n q^{\ell n - \binom{n}{2}} \, ,
    \end{align}
    we have 
    \begin{align}
        \sum_{n = 0}^\ell \frac{(-1)^{\ell-n} q^{\binom{\ell-n}{2}}}{(q;q)_{\ell-n} (q,b;q)_n} = \frac{(-1)^\ell q^{\binom{\ell}{2}}}{(q;q)_\ell} \sum_{n = 0}^\ell \frac{(q^{-\ell};q)_n}{(q,b;q)_n} q^n = \frac{(-1)^\ell q^{\binom{\ell}{2}}}{(q;q)_\ell} {}_2 \phi_1 
        \left( 
        \begin{matrix}
            q^{-\ell} , 0 \\ b
        \end{matrix} ; q, q
        \right) \, .
    \end{align}
    By the $q$-Chu--Vandermonde identity~\eqref{eq:q-Chu-Vandermonde}, we obtain 
    \begin{align}
        \frac{(-1)^\ell q^{\binom{\ell}{2}}}{(q;q)_\ell} {}_2 \phi_1 
        \left( 
        \begin{matrix}
            q^{-\ell} , 0 \\ b
        \end{matrix} ; q, q
        \right) = \frac{(-1)^\ell q^{\binom{\ell}{2}}}{(q;q)_\ell} \lim_{x \to 0}  {}_2 \phi_1 
        \left( 
        \begin{matrix}
            q^{-\ell} , x \\ b
        \end{matrix} ; q, q
        \right) = \frac{(-1)^\ell q^{\binom{\ell}{2}}}{(q;q)_\ell} \lim_{x \to 0} \frac{(b/x;q)_\ell}{(b;q)_\ell} x^\ell
        = \frac{b^\ell q^{2\binom{\ell}{2}}}{(q,b;q)_\ell} \, ,
    \end{align}
    so that
    \begin{align}
        C_d = \sum_{m=0}^d \frac{b^{d-m} q^{2\binom{d-m}{2}}}{(q,a;q)_m (q,b;q)_{d-m}} \, .
    \end{align}
    We may write
    \begin{align}
        C_d = \sum_{m=0}^d \frac{b^{m} q^{2\binom{m}{2}}}{(q,a;q)_{d-m} (q,b;q)_{m}} = \frac{1}{(q,a;q)_d} {}_2 \phi_1 
        \left( 
        \begin{matrix}
            q^{-d} , q^{1-d}/a \\ b
        \end{matrix} ; q, ab q^{2d-1}
        \right) \, .
    \end{align}
    Finally, applying the $q$-Chu--Vandermonde identity~\eqref{eq:q-Chu-Vandermonde2} again, we obtain 
    \begin{align}
        C_d = \frac{(ab q^{d-1};q)_d}{(q,a,b;q)_d} \, ,
    \end{align}
    which completes the proof.
\end{proof}

\subsection*{Schur polynomial}

A partition $\lambda = (\lambda_1,\lambda_2,\ldots)$ is a sequence of non-increasing non-negative integers, $\lambda_1 \ge \lambda_2 \ge \cdots \ge 0$.
We write $|\lambda| = \sum_{i \ge 1} \lambda_i$ and we denote the transposition of $\lambda$ by $\lambda^\text{T}$.
The length of $\lambda$ is $\ell(\lambda) = \lambda^\text{T}_1$, which is the number of non-zero elements in $\lambda$.

We define the Schur polynomial of $k$ variables as follows,
\begin{align}
    s_\lambda(z_1,\ldots,z_k) = \frac{\det_{1 \le i, j \le k} z_i^{\lambda_j + k - j}}{\det_{1 \le i, j \le k} z_i^{k - j}} = \frac{\det_{1 \le i, j \le k} z_i^{\lambda_{j}^\vee + j - 1}}{\det_{1 \le i, j \le k} z_i^{j - 1}} \, , \label{eq:Schur_def}
\end{align}
where $\lambda^\vee_j = \lambda_{k+1-j}$ is the inverse partition of $\lambda$.
For the $k$-variable case, $s_\lambda(z_1,\ldots,z_k) = 0$ when $\ell(\lambda) > k$.

% References

\bibliographystyle{ytamsalpha}
\bibliography{ref}

\end{document}